\documentclass[a4paper,UKenglish,cleveref, autoref, thm-restate]{lipics-v2021}

\title{Krenn-Gu conjecture for sparse graphs} %TODO 

\author{L. Sunil Chandran}{Indian Institute of Science, Bengaluru}{sunil@iisc.ac.in}{[orcid]}{Supported by SERB Core Research Grant CRG/2022/006770: ``Bridging Quantum Physics with Theoretical Computer Science and Graph Theory''}

\author{Rishikesh Gajjala}{Indian Institute of Science, Bengaluru}{rishikeshg@iisc.ac.in}{[orcid]}{}

\author{Abraham M. Illickan}{University of California, Irvine}{fill}{[orcid]}{}

\authorrunning{Chandran, Gajjala and Illickan} %TODO mandatory. First: Use abbreviated first/middle names. Second (only in severe cases): Use first author plus 'et al.'

\Copyright{L. Sunil Chandran, Rishikesh Gajjala, Abraham M. Illickan} %TODO mandatory, please use full first names. LIPIcs license is "CC-BY";  http://creativecommons.org/licenses/by/3.0/

 %TODO mandatory; please add comma-separated list of keywords

\category{} %optional, e.g. invited paper

\relatedversion{} 

\ArticleNo{1}

\usepackage{amsmath}
\usepackage{lipsum}
\usepackage{amsfonts}
\usepackage{graphicx}
\usepackage{epstopdf}
\usepackage{algorithmic}
\usepackage{tikz}
\usetikzlibrary{calc}
\newtheorem{subtheorem}{Theorem}
\nolinenumbers
\ifpdf
  \DeclareGraphicsExtensions{.eps,.pdf,.png,.jpg}
\else
  \DeclareGraphicsExtensions{.eps}
\fi

\usepackage{amsopn}

\ifpdf
\hypersetup{
  pdftitle={Resolution of Krenn-Gu conjecture for sparse graphs},
  pdfauthor={D. Doe, P. T. Frank, and J. E. Smith}
}
\fi

\begin{document}

\maketitle

\begin{abstract}

Greenberger–Horne–Zeilinger (GHZ) states are quantum states involving at least three entangled particles. They are of fundamental interest in quantum information theory, and the construction of such states of high dimension has various applications in quantum communication and cryptography. Krenn, Gu and Zeilinger discovered a correspondence between a large class of quantum optical experiments which produce GHZ states and edge-weighted edge-coloured multi-graphs with some special properties called the \emph{GHZ graphs}. On such GHZ graphs, a graph parameter called \emph{dimension} can be defined, which is the same as the dimension of the GHZ state produced by the corresponding experiment. Krenn and Gu conjectured that the dimension of any GHZ graph with more than $4$ vertices is at most $2$. An affirmative resolution of the Krenn-Gu conjecture has implications for quantum resource theory. Moreover, this would save huge computational resources used for finding experiments which lead to higher dimensional GHZ states. On the other hand, the construction of a GHZ graph on a large number of vertices with a high dimension would lead to breakthrough results.

In this paper, we study the existence of GHZ graphs from the perspective of the Krenn-Gu conjecture and show that the conjecture is true for graphs of vertex connectivity at most 2 and for cubic graphs. We also show that the minimal counterexample to the conjecture should be $4$-connected. Such information could be of great help in the search for GHZ graphs using existing tools like PyTheus. While the impact of the work is in quantum physics, the techniques in this paper are purely combinatorial, and no background in quantum physics is required to understand them.

\end{abstract}

\section{Introduction}
\label{sec:intro}
Quantum entanglement theory implies that two particles can influence each other, even though they are separated over large distances. In 1964, Bell demonstrated that quantum mechanics conflicts with our classical understanding of the world, which is local (i.e. information can be transmitted maximally with the speed of light) and realistic (i.e. properties exist prior to and independent of their measurement) \cite{bell}. Later, in 1989, Greenberger, Horne, and Zeilinger (abbreviated as GHZ) studied what would happen if more than two particles are entangled \cite{Greenberger}. Such states in which three particles are entangled (%denoted by
$|GHZ_{3,2}\rangle = \frac{1}{\sqrt{2}}\left(|000\rangle + |111\rangle \right)$) were observed rejecting local-realistic theories  ~\cite{PhysRevLett.82.1345,Pan2000}. 
While the study of such states started purely out of fundamental curiosity \cite{fund_cur1,fund_cur2,fund_cur3}, they are now used in many applications in quantum information theory, such as quantum computing \cite{Gu2020}.   They are also essential for early tests of quantum computing tasks \cite{quant_comp_tasks}, and quantum cryptography in quantum networks\cite{quant_networks}.

Zeilinger became a co-recipient of the Nobel Prize for Physics in 2022,
for experiments with entangled photons, establishing the violation of Bell inequalities and pioneering quantum information science.
We note that the work on experimentally constructing GHZ states is at the heart of  Zeilinger's  Nobel prize-winning work \cite{nobel}.
 Increasing the number of particles involved and the dimension of the GHZ state is essential both for foundational studies and practical applications. Motivated by this, a huge effort is being made by several experimental groups around the world to push the size of GHZ states. Photonic technology is one of the key technologies used to achieve this goal \cite{quant_comp_tasks,10photon}.
The Nobel Laureate himself, with some co-authors, proposed a scheme of optical experiments in order to
achieve this, which gives an opportunity for graph theorists to get involved in this fundamental
research: 
In 2017, Krenn, Gu and Zeilinger \cite{Quantum_graphs} discovered (and later extended \cite{Quantum_graphs_2,Quantum_graphs_3}) a bridge between experimental quantum optics and graph theory. They observed that large classes of quantum optics experiments (including those containing probabilistic photon pair sources, deterministic photon sources and linear optics elements) can be represented as an edge-coloured edge-weighted graph, though the edge-colouring goes a little beyond
what graph theorists are used to. Conversely, every edge-coloured edge-weighted graph (also referred to as an experiment graph) can be translated into a concrete experimental setup. This technique has led to the discovery of new quantum interference effects and connections to quantum computing \cite{Quantum_graphs_2}. Furthermore, it has been used as the representation of efficient AI-based design methods for new quantum experiments \cite{AI1,AI2}. %We note that while there are some experiments to construct GHZ states which does not fall into this category, these experiment graphs capture many prominent ones. %\textcolour{red}{Why GHZ graphs and connection}

However, despite several efforts, a way to generate a GHZ state of dimension $d>2$ with more than $n=4$ photons with perfect quality and finite count rates without additional resources \cite{krenn2019questions} could not be found. This led Krenn and Gu to conjecture that it is not possible to achieve this physically (stated in graph theoretic terms in \cref{krenn_conjecture}). They have also formulated this question purely in graph theoretic terms and publicised it widely among graph theorists for a resolution \cite{mixon_website}. We now formally state this problem in graph-theoretic terms and explain its equivalence in quantum photonic terms. For a high-level overview of how the experiments are converted to edge-coloured edge-weighted graphs, we refer the reader to the appendix of \cite{DBLP:journals/corr/abs-2202-05562}. For the exact details, the reader can refer to \cite{Quantum_graphs, Quantum_graphs_2, Quantum_graphs_3, krenn2019questions}.

\subsection{Graph theoretic preliminaries and notations}
We first define some commonly used graph-theoretic terms. For a graph $G$, let $V(G), E(G)$ denote the set of vertices and edges, respectively. We use $\kappa(G)$ to denote the vertex connectivity of $G$. For $S\subseteq V(G)$, $G[S]$ denotes the induced subgraph of $G$ on $S$. $\mathbb{N},\mathbb{N}_0,\mathbb{C}$ denote the set of natural numbers, non-negative numbers and complex numbers, respectively. The cardinality of a set $\cal{S}$ is denoted by $|\cal{S}|$. For a positive integer $r$, $[r]$ denotes the set $\{1,2\ldots,r\}$. Given a multi-graph, its skeleton is its underlying simple graph. We do not consider self-loops in multi-graphs.

Usually, in an edge colouring, each edge is associated with a natural number. However, in such edge colourings, the edges are assumed to be monochromatic. But in the graphs corresponding to experiments, we are allowed to have bichromatic edges, i.e. one half coloured by a certain colour and the other half coloured by a different colour. For example, in the graph shown in \cref{fig:main_example}, the simple edge between vertices $4$ and $6$ is a bichromatic edge. We develop some new notation to describe bichromatic edges. 

Each edge of a multi-graph can be thought to be formed by two half-edges, i.e., an edge $e$ between vertices $u$ and $v$, consists of the half-edge starting from the vertex $u$ to the middle of the edge $e$ (hereafter referred to as the $u$-half-edge of $e$) and the half-edge starting from the vertex $v$ to the
middle of the edge $e$ (hereafter referred to as the $v$-half-edge $e$). Thus, the edge set $E$ of the multi-graph gives rise to the set of half-edges $H$,
with $|H| = 2 |E|$.  For an edge $e$ between vertices $u$ and $v$, we may denote the $v$-half-edge of $e$ by $e_v$ and $u$-half edge of $e$ by $e_u$. Consider the edge $e$ between vertices $4$ and $6$ in \cref{fig:main_example}. The $4$-half edge of $e$ ($e_4$) is of colour red, and the $6$-half edge ($e_6$) is of colour green. %\textcolour{red}{fix description of edge}

The type of edge colouring that we consider in this paper is more aptly called a {\it half-edge colouring}. It is a function from  $H$ to  $\mathbb{N}_0$, say $c: H \rightarrow  \mathbb{N}_0$.  (Note that we use 
non-negative numbers to name the colours.) In other words, each half-edge gets a colour.   An edge is called monochromatic if both
its half-edges get the same colour (in which case we may use
$c(e)$ to denote this colour); otherwise, it is called a bi-chromatic edge. In \cref{fig:main_example}, the colour $0$ is shown in red, and the colour $1$ is shown in green. It is easy to see that $c(e_4)=0$ and $c(e_6)=1$ (recall that $e$ is the simple edge between vertices $4$ and $6$). Consider the edge $e'$ between vertices $1$ and $6$. As $c(e'_1)=c(e'_6)=0$, $e'$ is monochromatic and moreover, $c(e')=0$. We then assign a weight $w(e) \in \mathbb{C} $ to each such coloured edge $e$. We denote the multi-graph $G$ with the edge colouring $c$ and edge weights $w(e)$ as $G_c^w$.

We call a subset $P$ of edges in this edge-weighted edge-coloured graph a perfect matching if each vertex in the graph has exactly one edge in $P$ incident on it.
\begin{definition}
The weight of a perfect matching $P$, $w(P)$ is the product of the weights of all its edges $\prod\limits_{e\in P}w(e)$
\end{definition}

\begin{definition}
The weight of an edge-coloured edge-weighted multi-graph $G_c^w$ is the sum of the weights of all perfect matchings in $G_c^w$.
\end{definition}

A vertex colouring $vc$ associates a colour $i$ to each vertex in the graph for some $i\in \mathbb{N}$.
We use $vc(v)$ to denote the colour of vertex $v$ in the vertex colouring $vc$. A vertex colouring $vc$ \textit{filters out} a sub-graph $\mathcal{F}(G_c,vc)$ of $G_c$ on $V(G_c)$ where for an edge $e \in E(G_c)$ between vertices $u$ and $v$, $e \in E(\mathcal{F}(G_c,vc))$ if and only if $c(e_u)=vc(u)$ and $c(e_v)=vc(v)$. Filtering also extends to weighted graphs where the weight of each edge in $\mathcal{F}(G_c^w,vc)$ is the same as its weight in $G_c^w$. Let $vc$ be a vertex colouring in which $1,2,3,6$ are associated with the colour green and $4,5$ are associated with the colour red. The filtering operation of $vc$ on the edge-coloured graph $G_c^w$ shown in \cref{fig:main_example} is given in \cref{fig:main_example_filtered}. A vertex colouring $vc$ is defined to be feasible in $G_c^w$ if $\mathcal{F}(G_c^w,vc)$ has at least one perfect matching. It is easy to see that each perfect matching $P$ is part of $\mathcal{F}(G_c^w,vc)$ for a unique vertex colouring $vc$. Such a $P$ is said to induce $vc$. It is interesting to notice that there is a partition of perfect matchings (not edges) based on the vertex colourings.

\begin{definition}
The weight of a vertex colouring $vc$ in the multi-graph $G_c^w$ is denoted by $w(G_c^w,vc)$ and is equal to the weight of the graph $\mathcal{F}(G_c^w,vc)$.
\end{definition}

The weight of a vertex colouring, which is not feasible, is zero by default.

\begin{definition}
An edge-coloured edge-weighted graph is said to be \textit{GHZ}, if:
\begin{enumerate}
    \item All feasible monochromatic vertex colourings have a weight of $1$.
    \item All non-monochromatic vertex colourings have a weight of $0$.
\end{enumerate}
\end{definition}

\begin{figure}[t!]
    \centering
\centering    
\begin{subfigure}[b]{0.49\textwidth}
         \centering
         \fbox{\includegraphics[width=0.96\textwidth]{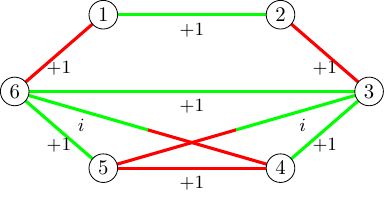}}
         \caption{An edge coloured edge-weighted graph $G_c^w$}
         \label{fig:main_example}
\end{subfigure}
\hfill
\begin{subfigure}[b]{0.49\textwidth}
         \centering
         \fbox{\includegraphics[width=0.96\textwidth]{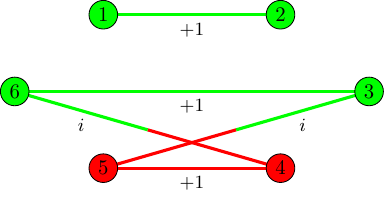}}
         \caption{$\mathcal{F}(G_c^w,vc)$, where  $vc(i)$ is green for $i\in [1,2,3,6]$ and red for $i\in [4,5]$}
         \label{fig:main_example_filtered}
\end{subfigure}
\hfill
\caption{$vc$-Filtering of a graph}
\label{fig:perfect_matchings}
\end{figure}
An example of a GHZ graph is shown in \cref{fig:main_example}.
\begin{definition}
The dimension of a GHZ graph $G_c^w$, $\mu(G,c,w)$ is the number of feasible monochromatic vertex colourings (having a weight of $1$).
\end{definition}

For a given multi-graph $G$ (experimental set up), many possible edge-colourings (mode numbers of photons) and edge-weight (amplitude of photon pairs) assignments may lead it a GHZ graph (GHZ state). Finding a GHZ graph with $n$ vertices and dimension $d$ would immediately lead to an experiment which result in a $d$-dimensional GHZ state with $n$ particles. For each such GHZ graph, a dimension is achieved. The maximum dimension achieved over all possible GHZ graphs with the unweighted uncoloured simple graph $G$ as their skeleton is known as the \textit{matching index} of $G$, denoted by $\mu(G)$. In \cref{fig:k2}, we have an edge-coloured edge-weighted GHZ $K_2$ of dimension $t$. Note that $t$ can be arbitrarily large. Therefore, $\mu(K_2)=\infty$. {However, only two particles are involved; for a GHZ state to form, we need more than two particles. Therefore, such a construction will not give a GHZ state.} We note that the matching index is defined for a simple graph $G$ by taking the maximum over all possible multi-graphs with a skeleton $G$. For instance, the simple graph $K_2$ has only one edge. However, we considered all possible multi-graphs having a skeleton $K_2$ to define $\mu(K_2)$. 

\begin{figure}[t!]
    \centering
\centering    
\begin{subfigure}[b]{0.4\textwidth}
         \centering
         {\includegraphics[width=40mm]{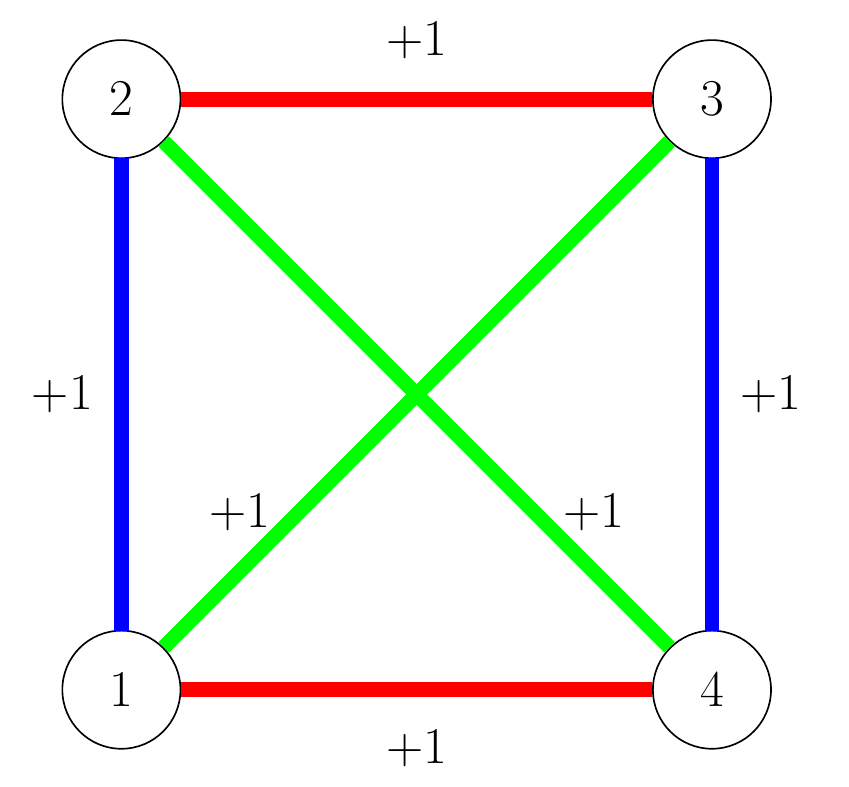}}
         \caption{$K_4$ with $3$ dimensions. %$\{1,2\},\{3,4\}$ are of colour blue. $\{1,3\},\{2,4\}$ are of colour green. $\{1,4\},\{2,3\}$ are of colour red.
         }
         \label{fig:3dimk4}
\end{subfigure}
\hfill
\begin{subfigure}[b]{0.5\textwidth}
         \centering
         {\includegraphics[width=60mm, height=40mm]{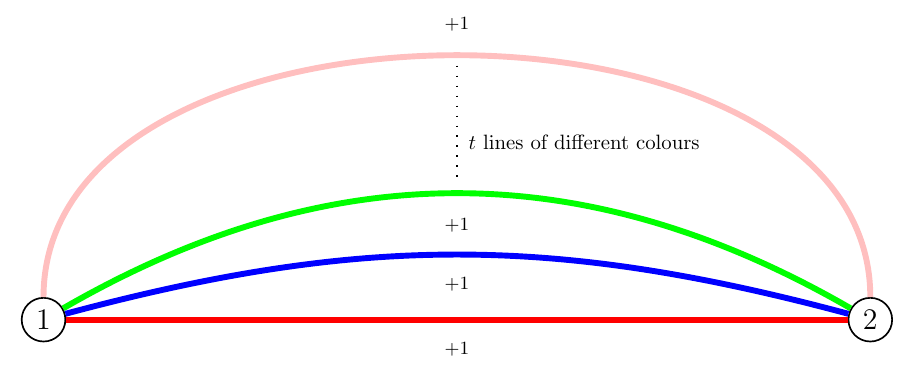}}
         \caption{$K_2$ with $t$ dimensions. }
         \label{fig:k2}
\end{subfigure}

\caption{GHZ graphs}
\label{fig:perfectly_mono_examples}
\end{figure}

It is easy to see that if a graph has a perfect matching, it must contain an even number of vertices. So, we consider matching indices of graphs with even and at least $4$ vertices for the rest of the manuscript. From \cref{fig:3dimk4}, we know that $\mu(K_4)\geq 3$ and, despite the use of huge computational resources \cite{AI1,AI2,neugebauer}, this is the only (up to an isomorphism) known graph of the matching index at least $3$. Any graph with a matching index of at least $3$ and $n>4$ vertices would lead to a new GHZ state of dimension at least $3$ with $n>4$ entangled particles. Motivated by this, this problem has been extensively promoted\cite{mixon_website,krenn_website}. Krenn and Gu conjectured that 
\begin{conjecture}
\label{krenn_conjecture}
If $|V(G)|>4$, then $\mu(G) \leq 2$ %and $\mu(K_4) = 3$.
\end{conjecture}
Several cash rewards were also announced for a resolution of this conjecture \cite{krenn_website}. We note the following implications of resolving this conjecture
\begin{enumerate}
    \item Finding a counterexample for this conjecture would uncover new peculiar quantum interference effects of a multi-photonic quantum system using which we can create new GHZ states
    \item
    \begin{enumerate}
    \item Proving this conjecture would immediately lead to new insights into resource theory in quantum optics
    \item Proving this conjecture for different graph classes would help us understand the properties of a counterexample and guide experimentalists in finding it. This is particularly important since huge computational efforts are going into finding such graphs \cite{AI1,AI2,neugebauer}. 
    \end{enumerate}
\end{enumerate}

A graph is matching covered if every edge of it is part of at least one perfect matching. If an edge $e$ is not part of any perfect matching $M$, then we call the edge $e$ to be redundant. By removing all redundant edges from the given graph $G$, we get its unique maximum matching covered sub-graph $mcg(G)$. Note that a colouring $c$ and a weight assignment $w$ of $G$ induces a colouring and a weight assignment for every subgraph of $G$, respectively. When there is no scope for confusion, we use $c,w$ itself to denote this induced colouring and weight assignment, respectively. It is easy to see that if $c$ and $w$ make $G$ a GHZ graph, they also make $mcg(G)$ a GHZ graph and $\mu(G,c,w)=\mu(mcg(G),c,w)$. Therefore, $\mu(G)=\mu(mcg(G))$.

One can also notice that, if there are two edges, say $e,e'$ between vertices $u$ and $v$ such that $c(e_u)=c(e'_u)$ and $c(e_v)=c(e'_v)$, then they can be replaced with an edge $e''$ such that $w(e'')=w(e)+w(e')$,$c(e''_u)=c(e_u)$ and $c(e''_v)=c(e_v)$. Such a reduction will retain the GHZ property and dimension of the graph. Therefore, in the rest of the manuscript, we only deal with such reduced graphs, i.e, between two vertices between vertices $u$ and $v$ and given $i,j \in [\mu(G)]$ there exists at most one edge $e$ such that $c(e_u)=i$ and $c(e_v)=j$. We also note that, if an edge $e$ has weight $0$, it can be treated as if the edge were absent.

% REQUIRED
%\begin{keywords}
% Graph colourings, Perfect matchings, Quantum Physics
%\end{keywords}

%\renewcommand{\thesubtheorem}{\thetheorem.\Alph{subtheorem}}
%\setcounter{theorem}{}

\subsection{Related work}

\textbf{No destructive interference.}
The special case of all edges having a real positive weight corresponds to the case when there is no destructive interference. With this restriction, Krenn-Gu conjecture was resolved due to the following observation by Bogdanov~\cite{bogdanov}. 
\begin{theorem}
    \label{bogdanov}
In a coloured multi-graph $G_c$ with $|V(G)|>4$, if there exist three monochromatic perfect matchings of different colours, then there must be a non-monochromatic perfect matching. 
\end{theorem}
Due to this result, when there is no destructive interference, every matching covered graph non-isomorphic to $K_4$ can achieve a maximum dimension of $1$ or $2$ and thus can be classified into Type $1$ and Type $2$ graphs(See \cite{DBLP:journals/corr/abs-2202-05562} for detailed discussion). Chandran and Gajjala~\cite{DBLP:journals/corr/abs-2202-05562} gave a structural classification for Type $2$ graphs. They further proved that for any half-edge colouring and edge weight assignment on a simple Type $2$ unweighted uncoloured graph, a dimension of $3$ or more can not be achieved! The computational aspects of the vertex colourings arising from these experiments were studied by Vardi and Zhang \cite{DBLP:journals/corr/abs-2209-13063, VardiZ23}

\textbf{Absence of bi-coloured edges.}
The problems get easier in the absence of bi-coloured edges and have opened up work in several directions. We list some of the known results in this direction. Cervera-Lierta et al.~\cite{Cervera} used SAT solvers to prove that if the number of vertices is $6$ or $8$, the maximum dimension achievable is $2$ or, at most, $3$, respectively. 
Chandran and Gajjala \cite{DBLP:journals/corr/abs-2304-06407} proved that the maximum dimension achievable for an $n>4$ vertex graph is less than $\dfrac{n}{\sqrt{2}}$. %\textcolour{red}{quantum paper 0.7n}

%\noindent
\textbf{Unrestricted results.} For the general case, the only known result is due to Mantey \cite{Kevin}. He proved the following theorem using the Gröbner basis.

\begin{theorem}
\label{kevin}
If $|V(G)|=4$, then $\mu(G)\leq 3$. Moreover, if $\mu(G,c,w)=3$, then between any pair of vertices in $G_c^w$, there is exactly one non-zero edge (and isomorphic to the coloured graph shown in \cref{fig:3dimk4}).
\end{theorem}
{Surprisingly, there is no known analytical proof even for such \textit{small} graphs. We encourage the reader to attempt to prove \cref{kevin} to understand the difficulty arising due to multi-edges. One has to tune $54$ variables which can be complex numbers (the number of possible edges when $3$ colours are allowed) such that $81$ equations (the number of possible vertex colourings when $3$ colours are allowed) are satisfied, even for graphs as small as $4$ vertex graphs.}
%We emphasize that a completely analytical proof for the fact that $\mu(K_4)=3$ is still missing and looks surprisingly challenging. Note that the difficulty arises due to multi-edges and bi-coloured edges. 

%\subsection{No destructive interference}

%When the weights of multiple perfect matchings corresponding to an inherited vertex colouring add up to zero, we say they are destructively interfering. When there is no destructive interference, the maximum dimension achievable in an experiment corresponding to a graph $G$ is defined as the graph's constructive matching index $\Bar{\mu}(G)$.  It may be noted that the case of all edge weights being positive real numbers corresponds to a special case when there is no destructive interference. This is because the weights of all perfect matchings are always positive and hence can not add up to zero.

\subsection{Our results}
We give the first results, which resolve the Krenn-Gu conjecture for a large class of graphs in the completely general setting, that is, when both bi-coloured edges and multi-edges are allowed.  We prove that the Krenn-Gu conjecture is true for all graphs with vertex connectivity at most $2$ in \cref{connectivity_proof}.

\begin{restatable}{theorem}{typetwores}
\label{type_2_resolution}
For a graph $G$, if $\kappa(G)\leq 2$, then $\mu(G)\leq 2$.
\end{restatable}

Our next main contribution is a reduction technique, which implies \cref{reduction_theorem_main}. We explain and prove our reduction in \cref{reduction}. We introduce a scaling lemma in \cref{generalization}, which gives us an equivalent version of Krenn-Gu conjecture and which may turn out to be more useful in some situations. 
\begin{theorem}
\label{reduction_theorem_main}
Given a graph $G$ with $\kappa(G) \leq 3$ and $V(G)>4$, there exists a graph $G'$ with $|V(G')|\leq |V(G)|-2$ and $\mu(G')\geq \mu(G)$.
\end{theorem}
Due to \cref{reduction_theorem_main}, a minimal counter-example (a counter-example with the minimum number of vertices) to Krenn-Gu conjecture must be $4$-connected. Using \cref{reduction_theorem_main}, we can resolve Krenn-Gu conjecture for some interesting graph classes like cubic graphs (that is, $3$ regular graphs). We prove \cref{result1} and \cref{result2} in \cref{red_appli}.
\begin{theorem}
\label{result1}
If the maximum degree of a graph $G$ is $3$, \cref{krenn_conjecture} is true.
\end{theorem}
%\textcolour{red}{add connections}
\begin{theorem}
\label{result2}
If the minimum degree of a graph $G$ is $3$, then $\mu(G)\leq 3$
\end{theorem}

\subsection{Reformulation of Krenn-Gu conjecture}
\label{generalization}
Recall that we denote the weight of the vertex colouring $vc$ over a set of vertices $U \subseteq V(G)$ as $w(U,vc)$, which is the sum of weights of all perfect matching on $G[U]$ which induce the vertex colouring $vc$ on $U$ and we denote the monochromatic vertex colouring $vc: V\to \{i\}$ by $\mathbf{i}_V$.
%\subsection{Reformulation}

For a graph $G$, let $U,U'\subseteq V(G)$. Let $vc:U\to \mathbf{N}$ and $vc':U'\to \mathbf{N}$. If $vc(v)=vc'(v)$ for all $v \in U\cap U'$, we call $vc,vc'$ to be compatible with each other. When $vc,vc'$ are compatible, we define their union $[vc \cup vc']:U\cup U' \to \mathbf{N}$ as follows: $[vc \cup vc'](v)=vc(v)$ for $v\in U$ and $[vc \cup vc'](v)=vc'(v)$ for $v\in U'$. 

We broaden the definition of GHZ graphs to \textit{g-GHZ} graphs. An edge-coloured edge-weighted graph $G_c^w$ satisfying the following properties is defined to be g-GHZ
\begin{enumerate}
    \item All feasible monochromatic vertex colourings have a non-zero weight (instead of necessarily being $1$).
    \item All non-monochromatic vertex colourings have a weight of $0$.
\end{enumerate}
Note that this generalization allows each of the monochromatic vertex colourings to have different weights. The dimension of a g-GHZ graph is the number of feasible monochromatic vertex colourings. For a graph $G$, the maximum dimension achievable over all possible g-GHZ colouring and weight assignments is its \textit{generalized matching index} $\mu_g(G)$. 
\begin{conjecture}
\label{krenns_conjecture_generalization}
$\mu_g(K_4) = 3$ and for a graph $G$ which is non-isomorphic to $K_4$, $\mu_g(G) \leq 2$.
\end{conjecture}
%\textcolour{red}{numbering}
We prove that \cref{krenn_conjecture} and \cref{krenns_conjecture_generalization} are equivalent. Trivially, a counter-example to \cref{krenn_conjecture} would immediately give a counter-example to \cref{krenns_conjecture_generalization}. We prove that any counter-example to \cref{krenns_conjecture_generalization}  would also yield a counter-example to \cref{krenn_conjecture} in \cref{scaling_generalization}. This reformulation is more suitable for {our proofs in the following sections.}
\begin{lemma}[Scaling lemma]
\label{scaling_generalization}
If there is a graph $G_c^w$, which is g-GHZ, then there is a graph $G_c^{w'}$, which is a GHZ graph with the same dimension.
\end{lemma}

\begin{proof}
    We denote the weight of the vertex colouring $w(\mathbf{i}_V, G)$ using $W(i)$. Note that by definition of g-GHZ graphs, the weight of a monochromatic colouring $W(i)$ is always non-zero. So, for each edge $e\in G_c$ whose half-edges are of colour $i,j$, we assign the weight 

$$w'(e)=w(e)(W(i)W(j))^{-1/n}$$

Let $M$ be a matching in $G_c$, which induces the vertex colouring $vc_M$. The weight of an edge $e \in M$ between vertices $u,v$ will be, 
$$w'(e)=w(e)(W(vc_M(u))W(vc_M(v)))^{-1/n}$$ 
As each vertex is incident by exactly one edge of the perfect matching $M$, the weight of M will be $$w'(M)=w(M)\prod\limits_{v \in V} W(vc_M(v))^{-1/n}$$ 
Since the weights of all perfect matchings which induce a vertex colouring, $vc$ will increase by a factor of $\prod\limits_{v \in V} W(vc(v))^{-1/n}$, the weight of the vertex colouring $vc$ will be
$$w'(vc)=w(vc) \prod\limits_{v \in V} W(vc(v))^{-1/n} $$   
As $w(vc)$ is zero for all non-monochromatic vertex colourings, $w'(vc)$ will remain to be zero. 

For a monochromatic vertex colouring $vc=\mathbf{i}_V$, we know that  $vc(v)=i$ for all $v\in V$. Therefore,  $w'(\mathbf{i}_V)=w(\mathbf{i}_V) ((w(\mathbf{i}_V))^{-1/n})^n=1 $

Therefore, $G_c^{w'}$ is a GHZ graph.
\end{proof}
\section{Matching index of multigraphs with vertex connectivity at most 2}\label{connectivity_proof}

In a graph $G$ with $V(G)=V_1\bigsqcup V_2 \cdots \bigsqcup V_k$, the vertex colouring of $G$ in which each $V_i$ receives the colour $i$ is denoted by $\mathbf{1}_{V_1}\mathbf{2}_{V_2}\cdots \mathbf{k}_{V_k}$. The weight of the vertex colouring $\mathbf{1}_{V_1}\mathbf{2}_{V_2}\cdots \mathbf{k}_{V_k}$ in a graph $G_{c'}^w$ is denoted as $w(G_{c'}^w,\mathbf{1}_{V_1}\mathbf{2}_{V_2}\cdots \mathbf{k}_{V_k})$.  When the graph $G_{c'}^w$ is understood from the context,  we may shorten
the notation to just $w(\mathbf{1}_{V_1}\mathbf{2}_{V_2}\cdots \mathbf{k}_{V_k})$. We also use the notation $w(\mathbf{i}_{V_i}\mathbf{j}_{V_j}\mathbf{k}_{V_k})$ to denote be the weight of the vertex colouring $\mathbf{i}_{V_i}\mathbf{j}_{V_j}\mathbf{k}_{V_k}$ on the induced subgraph of $V_i \bigsqcup V_j \bigsqcup V_k$.

%Consider the induced subgraph on the vertices $V_i \bigsqcup V_j \bigsqcup V_k$

%Also, for the given vertex colouring of $G_c^w$, the weight of the induced subgraph on a few colour classes, say $V_i \bigsqcup V_j \bigsqcup V_k$, with respect to $\mathbf{i}_{V_i}\mathbf{j}_{V_j}\mathbf{k}_{V_k}$ is $w(\mathbf{i}_{V_i}\mathbf{j}_{V_j}\mathbf{k}_{V_k})$.

%The vertex colouring of a graph in which the disjoint vertex subsets $A,B,\dots$ receive the colours $i,j,\dots$ respectively is denoted as $\mathbf{i}_A\mathbf{j}_B\dots$. The weight of the vertex colouring $\mathbf{i}_A\mathbf{j}_B\dots$ in a graph $G$ is denoted as $w(G,\mathbf{i}_A\mathbf{j}_B\dots)$. 

\typetwores*
\begin{proof}

We will consider only matching covered graphs. If a graph is not matching covered, remove all redundant edges to get $G'=mcg(G)$. Since the removed edges were not part of any perfect matchings, $\mu(G)=\mu(G')$. The hypothesis is still true since $\kappa(G')\leq\kappa(G)$.

If the graph has no perfect matching, then ${\mu}(G)=0$. So, we may assume that there is at least one perfect matching (${\mu}(G)\geq 1$). Note that since there is a perfect matching, the number of vertices is even.

\begin{observation}\label{conn0obs}
If $\kappa(G)=0$, then ${\mu}(G) \leq 1$. 
\end{observation}\label{disconnected}
\begin{proof}
If $\kappa(G)=0$, $G$ is disconnected. Assume, towards a contradiction, that ${\mu}(G)\geq 2$. Let $1$ and $2$ be two colours whose monochromatic vertex colourings have a weight of $1$. Let $A$ be the set of vertices in one of the connected components and $B$ be the remaining vertices $V(G)\setminus A$.

Since $A$ and $B$ are disconnected, the weight of monochromatic vertex colouring $\mathbf{1}_A\mathbf{1}_B$ of $G$ is
\begin{equation}\label{k0A_1}
   1= w(G,\mathbf{1}_A\mathbf{1}_B) =w(G[A],\mathbf{1}_A)w(G[B],\mathbf{1}_B)\implies w(G[A],\mathbf{1}_A)\neq 0    
\end{equation}
Similarly, the weight of monochromatic vertex colouring $\mathbf{2}_A\mathbf{2}_B$ of $G$ is
\begin{equation}\label{k0B_2}
 1=w(G,\mathbf{2}_A\mathbf{2}_B)= w(G[A],\mathbf{2}_A)w(G[B],\mathbf{2}_B)\implies w(G[A],\mathbf{2}_B)\neq 0
\end{equation}
From \cref{k0A_1} and \cref{k0B_2},
\begin{equation*}
    w(G[A],\mathbf{1}_A)w(G[B],\mathbf{2}_B)\neq 0
\end{equation*}
Consider the vertex colouring $\mathbf{1}_A\mathbf{2}_B$. This is a non-monochromatic vertex colouring, and its weight should be $0$.
\begin{equation*}
0 = w(G,\mathbf{1}_A\mathbf{2}_B)  = w(G[A],\mathbf{1}_A)w(G[B],\mathbf{2}_B)
\end{equation*}
This is a contradiction.
\end{proof}

\begin{observation}\label{conn1obs}
For a matching covered graph $G$, $\kappa(G)\neq 1$.
\end{observation}

\begin{proof}
Suppose not. If $\kappa(G)=1$, there is a cut vertex $v$. Since the total number of vertices in $G$ is even, the number of vertices in $V(G)-\{v\}$ is odd. The removal of $v$ separates the graph into two or more components. As the total number of vertices in $V(G)-\{v\}$ is odd, there must be at least one connected component in the graph induced on $V(G)-\{v\}$ with an odd number of vertices, say $A$. Let $B=V(G)-\{v\}-A$. Note that $|B|$ must be even. As $|A|$ is odd, in any perfect matching in $G$, there will necessarily be an edge from $A$ to $v$ and hence no edge from $B$ to $v$. Therefore, as $G$ is a matching covered graph, $A\cup\{v\}$ must be disconnected from $B$. Therefore, $\kappa(G)=0$. This is a contradiction. 
\end{proof}

We will now proceed to prove that if $\kappa(G) = 2$, then ${\mu}(G)\leq 2$. As $\kappa(G) = 2$, there is a set $U=\{u,v\}$ which separates the graph. 

Let $A$ be the set of vertices in a connected component of the induced subgraph of $V(G)-U$ on $G$ and $B$ be the remaining vertices in $V(G)-U-A$. We now consider two cases depending on whether the number of vertices $A$ is odd or even. The reader may note that the proof for the case when the number of vertices in $A$ is even is almost similar to the case when the number of vertices in $A$ is odd, except for some subtleties which are presented in the latter case.

\begin{lemma}\label{oddlemma}
If $|A|$ is odd, then ${\mu}(G)\leq 2$.
\end{lemma}

\begin{proof}
If the number of vertices in $A$ is odd, in every perfect matching, there will be a vertex in $A$ which is matched with a vertex in $\{u,v\}$. So, we can divide the perfect matchings in $G$ into two types depending on whether such vertex is matched with $u$ or $v$.

\begin{enumerate}
    \item\label{oddtype1} Perfect matchings that have an edge from $A$ to $u$ and an edge from $B$ to $v$
    \item\label{oddtype2} Perfect matchings that have an edge from $A$ to $v$ and an edge from $B$ to $u$
\end{enumerate}

Let us denote $G[A\cup \{u\}]$, $G[A\cup \{v\}]$, $G[B\cup \{u\}]$, $G[B\cup \{v\}]$, with $G_{Au}$, $G_{Av}$, $G_{Bu}$, $G_{Bv}$ respectively.

Consider the vertex colouring $c=\mathbf{i}_A\mathbf{j}_B\mathbf{k}_{\{u\}}\boldsymbol{\ell}_{\{v\}}$. As $A$ is an odd-sized set, every perfect matching in $G$ inducing $c$ would fall into exactly one of the two types described above. It is easy to see that the perfect matchings in $G$ of Type \ref{oddtype1}, which induce the vertex colouring $c$, are the disjoint unions of the perfect matchings in $G_{Au}$ which induce the vertex colouring $\mathbf{i}_A\mathbf{k}_{\{u\}}$ and the perfect matchings in $G_{Bv}$ which induce the vertex colouring $\mathbf{j}_B\boldsymbol{\ell}_{\{v\}}$.  Therefore, the contribution to the weight of $w(c)$ due to perfect matchings of Type
\ref {oddtype1}  is the product of $w(i_Ak_{\{u\}})$ and $w(j_B\ell_{\{v\}})$. 

Similarly, the perfect matchings in $G$ of Type \ref{oddtype2} which induce the vertex colouring $c$ are the disjoint unions of the perfect matchings in $G_{Av}$ which induce the vertex colouring $\mathbf{i}_A\boldsymbol{\ell}_{\{v\}}$ and the perfect matchings in $G_{Bu}$ which induce the vertex colouring $\mathbf{j}_B\mathbf{k}_{\{u\}}$. The contribution to $w(c)$ due to perfect matchings
of Type \ref {oddtype2} equals the product of $w(\mathbf{i}_A\boldsymbol{\ell}_{\{v\}})$ and $w(\mathbf{j}_B\mathbf{k}_{\{u\}})$.

From here on, with a slight abuse of notation, we will use the small letters $u$ and $v$ to represent the singleton sets $\{u\}$ and $\{v\}$, respectively.

\begin{figure}[htbp]
\centering
\tikzset{
    vertex/.style={circle, draw, minimum size=5pt, inner sep=0pt, fill=black}
}
\begin{tikzpicture}
% Define vertices
\node[vertex,label=below:$A$] (A) at (0,0) {};
\node[vertex,label=above right:{$B$}] (B) at (2,2) {};
\node[vertex,label=above left:{$u$}] (u) at (0,2) {};
\node[vertex,label=below right:{$v$}] (v) at (2,0) {};

% Draw edges
\draw (A) -- node[left] {${\cal V}_{c}^l$} (u);
\draw (A) -- node[below] {$H_{c}^b$} (v);
\draw (B) -- node[right] {${\cal V}_{c}^r$} (v);
\draw (B) -- node[above] {$H_{c}^t$} (u);
\end{tikzpicture}
\caption{Schematic diagram for \cref{oddlemma}}
\label{fig:graph_ex}
\end{figure}
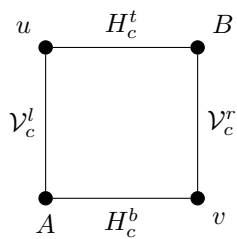

Since this case is more sophisticated than the earlier cases,  to make it possible to visualize 
some steps of the proof, we introduce a schematic diagram to represent the weight of a colouring 
$i_Aj_Bk_u\ell_v$ of $G$.   We consider that  the square shown in figure \ref{fig:graph_ex} to represent the graph,
with the bottom left-hand corner (vertex) representing the component $A$,  the top right-hand corner (vertex)  representing $B$, the top left-hand corner  (vertex) and bottom right-hand corner (vertex) representing $u$ and $v$ 
respectively.  An edge of the square represents the corresponding induced subgraph of $G$; for example,
the left vertical edge represents the induced subgraph $G_{Au}$. If a colouring, say $i_Aj_Bk_u\ell_v$ is given,  we may  
represent this colouring of $G$ by a coloured square where the vertices that correspond to $A, B,u$ and $v$ gets the colours $i,j,k$ and $\ell$ respectively. 
In a coloured square, we may consider that the edges represent the weight of the induced
subgraph that corresponds to it with respect to this colouring. Thus, for the coloured square of
the colouring $i_Aj_Bk_u\ell_v$,  the top horizontal edge represents the weight of the colouring $k_uj_B$ of $G_{Bu}$.

For a given vertex colouring $c= i_Aj_Bk_u\ell_v$, the  product of  the two (i.e. top and bottom) horizontal  edges of the coloured square is called the  horizontal weight of the colouring and is denoted by
$H_c$.  The top horizontal edge and the bottom horizontal edge are represented by
$H_c^t$ and $H_c^b$ respectively. Thus $H_c = H_c^t.H_c^b$. Similarly, the product of
 the two vertical edges (left vertical edge ${\cal V}_c^l$ and right vertical edge ${\cal V}_c^r$)
 of the coloured square of $c$
 is the vertical weight of $c$, denoted by ${\cal V}_c = {\cal V}_c^l. {\cal V}_c^r$.  The reader may convince
 herself that $H_c$ and ${\cal V}_c$ represents the contribution to $w(c)$ due to Type \ref {oddtype2}
and Type \ref {oddtype1}  perfect matchings respectively. 
Since $w(c)$ is the sum of the weights due to perfect matchings of
type \ref {oddtype1} and type \ref {oddtype2},  it is easy to see that  $w(c) = H_c + {\cal V}_c$.

\noindent Since for any non-monochromatic coloring $c$, $w(c) = 0$ the following is easy to see.

\begin{observation}\label{k2oddnonmono0}
    If $i$, $j$, $k$, $\ell$ are not all equal, for $c=i_Aj_Bk_u\ell_v$,
      $H_c = -{\cal V}_c$
\end{observation}

For ease of discussion, we will call a coloured square a \emph {solid coloured square}  if 
\emph {all}  its edges are non-zero and a \emph {fragile colored square} otherwise.

\begin {observation}
\label {nonmonochromatic_corollary}
For a non-monochromatic coloring $c$, if either of $H_c$ or ${\cal V}_c$ is non-zero, then all the $4$ edges of the
colored square of $c$ represent non-zero weights, i.e. $H_c^t, H_c^b, {\cal V}_c^l, {\cal V}_c^r \ne 0$. In other words, the coloured square of $c$ would be solid. 
\end {observation}

%We will now proceed to show that if the matching index is at least 3, then none of these terms in \cref {nonzero_terms_observation} can be zero. In other words, the coloured square of $\delta = 2_A2_B1_u1_v$ is indeed solid,  which is a contradiction, and therefore, the number of vertices in $A$ can not be odd. 
 Since ${\mu}(G)\geq 3$, there exist at least $3$ monochromatic vertex colourings whose weight is $1$. Let $c_i = (i_A,i_B,i_u, i_v)$ for $i=1,2,3$.
Clearly for each $i \in [3]$, we have $w(c_i) = 1$, and therefore either  (i.e. at least one
of) 
$H_{c_i}$ or ${\cal V}_{c_i}$ is non-zero.  In the former case, let $c_i$ be in Class 1; in the
latter case, let $c_i$ be in  Class 2. By pigeon hole principle, at least $2$ colours 
from $c_i, i \in [3]$ should belong to the same class. Without loss of generality, 
let these two colours be $c_1$ and $c_2$ and moreover, without loss of generality, let us assume that
they belong to Class 2. 

\begin{observation}\label{zero_term_obs}
 If $\mu(G)\geq 3$, %then there exists colours $1$ and $2$ such that 
 $$ w(\mathbf{2}_A\mathbf{1}_{\{u\}})w(\mathbf{2}_B\mathbf{1}_{\{v\}})w(\mathbf{2}_A\mathbf{1}_{\{v\}})w(\mathbf{2}_B\mathbf{1}_{\{u\}})\neq 0$$   
\end{observation}

\begin{proof}

\begin{figure}[htbp] % "htbp" specifies the preferred placement of the figure (here: "here", "top", "bottom", "page")
  \centering % Center the figure
  \includegraphics[width=0.999\textwidth]{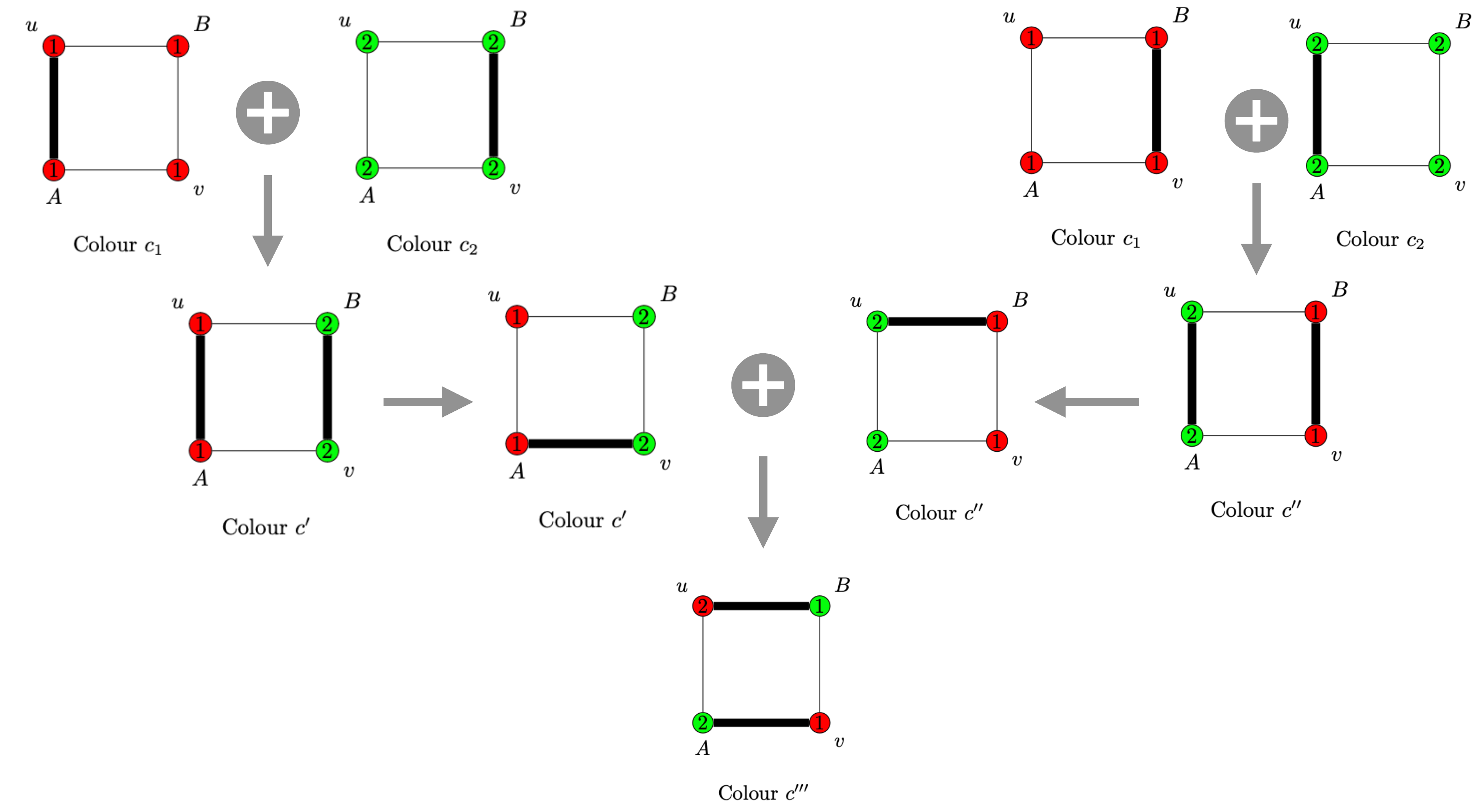} % Include the image and specify its width
  \caption{Figure illustrating the equations of \cref{zero_term_obs}.} % Caption for the figure
  \label{fig:example_zero_term_obs} % Label for referencing the figure
\end{figure}

By corollary \ref {nonmonochromatic_corollary}, it is 
enough to show that either $H_\gamma$ or ${\cal V}_\gamma$ is non-zero for the colored square of $\gamma=\mathbf{2}_A\mathbf{2}_B\mathbf{1}_u\mathbf{1}_v$, to show
that it is solid. 

As $c_1$ and $c_2$ belong to Class 2 by assumption,  we know that ${\cal V}_{c_1}, {\cal V}_{c_2} \ne 0$ and therefore 
${\cal V}_{c_1}^l, {\cal V}_{c_1}^r, {\cal V}_{c_2}^l, {\cal V}_{c_2}^r$ are all non-zero. Therefore, when $c' = 1_A2_B1_u2_v$
$$0 \ne {\cal V}_{c_1}^l. {\cal V}_{c_2}^r = w(1_A,1_u).w(2_B,2_u)=  {\cal V} (c')$$  
When $c'' = 2_A 1_B 2_u 1_v$, in the same way, since  ${\cal V}_{c_1}^r, {\cal V}_{c_2}^l \ne 0$, we get  $$0 \ne   {\cal V}_{c_1}^r. {\cal V}_{c_2}^l  = w(1_B1_v).w(2_A2_u)={\cal V} (c'') $$ 

Applying \cref {nonmonochromatic_corollary} on the coloring $c'$ and $c''$, we can infer that square$(c')$  and
square$(c'')$  are both  solid and hence $H_{c'}^t, H_{c''}^b \ne 0$. Therefore, when $c'''=2_A2_B1_u1_v$
$$0 \neq H_{c'}^t H_{c''}^b = w(2_B1_u).w(2_A1_v)= H_{c'''}^t H_{c'''}^b=H(c''')$$
Applying \cref {k2oddnonmono0}, we get that $V(c''')\neq 0$. Therefore, $c'''$ is solid. It now follows that $$ w(\mathbf{2}_A\mathbf{1}_{\{u\}}),w(\mathbf{2}_B\mathbf{1}_{\{v\}}),w(\mathbf{2}_A\mathbf{1}_{\{v\}}),w(\mathbf{2}_B\mathbf{1}_{\{u\}})\neq 0$$
Our proof is illustated in \cref{fig:example_zero_term_obs}.
%This is illustrated in \cref{fig:graph_exxx}. 
\end{proof}

We will now prove that $c'''$ is fragile and thus contradicting \cref{zero_term_obs}.
\begin{observation}
\label {nonzero_terms_observation}
\begin{equation*}
 w(\mathbf{2}_A\mathbf{1}_{\{u\}})w(\mathbf{2}_B\mathbf{1}_{\{v\}})w(\mathbf{2}_A\mathbf{1}_{\{v\}})w(\mathbf{2}_B\mathbf{1}_{\{u\}})=0
 \end{equation*}
 In other words, the coloured square of the colouring $\mathbf{2}_A\mathbf{2}_B\mathbf{1}_u\mathbf{1}_v$ is fragile.
\end{observation}

\begin{proof}

Consider the following four vertex colourings of $G$:  $\alpha = \mathbf{2}_A\mathbf{2}_B\mathbf{1}_u\mathbf{2}_v$, $\beta = 
 \mathbf{2}_A\mathbf{2}_B\mathbf{2}_u\mathbf{1}_v$, $\gamma=\mathbf{2}_A\mathbf{2}_B\mathbf{1}_u\mathbf{1}_v$ and $\delta=\mathbf{2}_A\mathbf{2}_B\mathbf{2}_u\mathbf{2}_v$. 

 From observation \cref {k2oddnonmono0}  it is clear that 
  ${\cal V}_\alpha.{\cal V}_\beta. H_\gamma    = -H_\alpha.H_\beta.{\cal V}_\gamma $.   This expands to 

  $$   
{\cal V}_\alpha^l.{\cal V}_\alpha^r.{\cal V}_\beta^l.{\cal V}_\beta^r.H_\gamma^b.H_\gamma^t = -      H_\alpha^b.H_\alpha^t.H_\beta^b H_\beta^t.{\cal V}_\gamma^l.{\cal V}_\gamma^r  $$

%Observe that $H_\alpha^t=H_\gamma^t$, $H_\beta^b=H_\gamma^b$, ${\cal V}_\gamma^l = {\cal V}_\alpha^l$ and ${\cal V}_\gamma^r={\cal V}_\alpha^r$. Therefore, 

%$$   {\cal V}_\alpha^l.{\cal V}_\alpha^r.{\cal V}_\beta^l.{\cal V}_\beta^r.H_\gamma^b.H_\gamma^t = -      H_\alpha^b.H_\gamma^t.H_\gamma^b.H_\beta^t.{\cal V}_\alpha^l.{\cal V}_\alpha^r $$

%$$ \implies {\cal V}_\alpha^l.{\cal V}_\alpha^r.H_\gamma^b.H_\gamma^t({\cal V}_\beta^l.{\cal V}_\beta^r + H_\alpha^bH_\beta^t)=0 $$
%Observe that 
Now, substituting the original weights for each term above, we get the following:
\begin{multline*}
    w(\mathbf{2}_A\mathbf{1}_{\{u\}})w(\mathbf{2}_B\mathbf{2}_{\{v\}})w(\mathbf{2}_A\mathbf{2}_{\{u\}})w(\mathbf{2}_B\mathbf{1}_{\{v\}})w(\mathbf{2}_A\mathbf{1}_{\{v\}})w(\mathbf{2}_B\mathbf{1}_{\{u\}})\\
    = -w(\mathbf{2}_A\mathbf{2}_{\{v\}})w(\mathbf{2}_B\mathbf{1}_{\{u\}})w(\mathbf{2}_A\mathbf{1}_{\{v\}})w(\mathbf{2}_B\mathbf{2}_{\{u\}})w(\mathbf{2}_A\mathbf{1}_{\{u\}})w(\mathbf{2}_B\mathbf{1}_{\{v\}})
\end{multline*}
It follows that 
\begin{multline*}
    (w(G_{Au},\mathbf{2}_A\mathbf{2}_{\{u\}})w(G_{Bv},\mathbf{2}_B\mathbf{2}_{\{v\}}) + w(G_{Av},\mathbf{2}_A\mathbf{2}_{\{v\}})w(G_{Bu},\mathbf{2}_B\mathbf{2}_{\{u\}}))\\ 
    \times   w(G_{Au},\mathbf{2}_A\mathbf{1}_{\{u\}})w(G_{Bv},\mathbf{2}_B\mathbf{1}_{\{v\}})w(G_{Av},\mathbf{2}_A\mathbf{1}_{\{v\}})w(G_{Bu},\mathbf{2}_B\mathbf{1}_{\{u\}}) = 0
\end{multline*} 
%\textcolor{red}{proof environment}

As $\delta$ is monochromatic, we get $w(G,\delta)=w(G,\mathbf{2}_A\mathbf{2}_B\mathbf{2}_u\mathbf{2}_v)=1 \neq 0$. It follows that
$$ 0\neq w(G_{Au},\mathbf{2}_A\mathbf{2}_{\{u\}})w(G_{Bv},\mathbf{2}_B\mathbf{2}_{\{v\}}) +w(G_{Av},\mathbf{2}_A\mathbf{2}_{\{v\}})w(G_{Bu},\mathbf{2}_B\mathbf{2}_{\{u\}}) $$

Therefore,
$$w(G_{Au},\mathbf{2}_A\mathbf{1}_{\{u\}})w(G_{Bv},\mathbf{2}_B\mathbf{1}_{\{v\}})w(G_{Av},\mathbf{2}_A\mathbf{1}_{\{v\}})w(G_{Bu},\mathbf{2}_B\mathbf{1}_{\{u\}}) = 0$$
Therefore, the coloured square of the colouring $\mathbf{2}_A\mathbf{2}_B\mathbf{1}_u\mathbf{1}_v$ is fragile. 
\end{proof}

Hence, from \cref{zero_term_obs} and \cref {nonzero_terms_observation}, if $|A|$ is odd, then ${\mu}(G)\leq 2$.
\end{proof}

\begin{lemma}\label{evenlemma}
If $|A|$ is even, then $\mu(G)\leq 2$    
\end{lemma}
\begin{proof}
If the number of vertices in $A$ is even, we can divide the perfect matchings in $G$ into two types.

\begin{enumerate}
    \item\label{eventype1} Perfect matchings that have an edge from $A$ to both $u$ and $v$ and no edge from $B$ to $U$
    \item\label{eventype2} Perfect matchings that have no edge from $A$ to $U$
\end{enumerate}
Note that the perfect matchings in which the edge $(u,v)$ is present are of the Type \ref{eventype2}. Let us denote $G[A\cup U]\setminus E({u,v})$, $G[A]$, $G[B\cup U]$, $G[B]$, with $G'_{AU}$, $G_{A}$, $G_{BU}$, $G_{B}$ respectively.

Consider the vertex colouring $c=\mathbf{i}_A\mathbf{j}_B\mathbf{k}_{\{U\}}$. As $A$ is an even-sized set, every perfect matching in $G$ inducing $c$ would fall into exactly one of the two types described above. The perfect matchings in $G$ of Type \ref{eventype1} which induce the vertex colouring $c$ are the disjoint unions of the perfect matchings in $G'_{AU}$ which induce the vertex colouring $\mathbf{i}_A\mathbf{k}_{U}$ and the perfect matchings in $G_{B}$ which induce the vertex colouring $\mathbf{j}_B$. The perfect matchings in $G$ of Type \ref{eventype2} which induce the vertex colouring $c$ are precisely the disjoint unions of the perfect matchings in $G_{A}$ which induce the vertex colouring $\mathbf{i}_A$ and the perfect matchings in $G_{BU}$ which induce the vertex colouring $\mathbf{j}_B\mathbf{k}_{U}$.

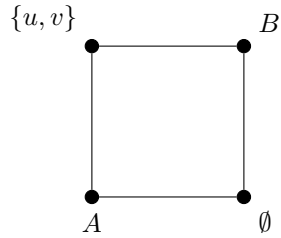
\begin{figure}[htbp]
\centering
\tikzset{
    vertex/.style={circle, draw, minimum size=5pt, inner sep=0pt, fill=black}
}
\begin{tikzpicture}
% Define vertices
\node[vertex,label=below:$A$] (A) at (0,0) {};
\node[vertex,label=above right:{$B$}] (B) at (2,2) {};
\node[vertex,label=above left:{$\{u,v\}$}] (u) at (0,2) {};
\node[vertex,label=below right:{$\emptyset$}] (v) at (2,0) {};

% Draw edges
\draw (A) -- (u);
\draw (A) -- (v);
\draw (B) -- (v);
\draw (B) -- (u);
\end{tikzpicture}
\caption{Schematic diagram for \cref{evenlemma}}
\label{fig:graph_ex2}
\end{figure}

We now create a new schematic diagram in \cref{fig:graph_ex2}. There are some subtle differences between our new schematic diagram and new schematic diagram the from \cref{fig:graph_ex}. Even in this case, we can visualize the weight of a colouring $c$ using the schematic diagram of a coloured square, square$(c)$. Just that the top left-hand corner
of the square, in this case, represents the set $\{u,v\}$ and the bottom right-hand corner represents $\emptyset$.  Like in the earlier
case, the bottom left-hand corner and top right-hand corner represent $A$ and $B$, respectively. Note that we will consider only colourings $c$ such that that the induced colourings on $A, B$ and $U=\{u,v\}$ are monochromatic, for example, $i_Aj_Bk_U$.  

In a coloured square, we may consider that the edges represent the weight of the induced
subgraph that corresponds to it with respect to this colouring. Thus, for a vertex colouring $i_Aj_Bk_U$, the weight of top horizontal edge represents the weight of the colouring $k_Uj_B$ of $G_{BU}$ and the weight of the left vertical edge represents the weight of the colouring $i_Ak_U$ of $G_{AU}$. Clearly, we have $w(c) = H_c + {\cal V}_c$ for any vertex colouring $c$ like earlier. Since for non-monochromatic colourings $w(c) = 0$, we have the following.  

%The remaining part of the proof for this case is similar to the case when the number of vertices is odd. We replicate it with some changes involving $G'_{AU}$.

\begin{observation}\label{k2evennonmono0}
If $i$, $j$, $k$ are not all equal, for vertex coloring $c=i_Aj_Bk_U$
\begin{equation*}
    H_{c}=-{\cal V}_c
\end{equation*}
\end{observation}

\begin {observation}
\label {even_solid_square}
For a non-monochromatic colouring $c$, if either $H_c$ or ${\cal V}_c$ is non-zero, then the square$(c)$ is solid. 
\end {observation}

As ${\mu}(G)\geq 3$, there exist at least $3$ monochromatic vertex colourings whose weight is $1$. Going via the pigeonhole principle based argument as in the proof of \cref{oddlemma}, we infer that at least 2 of these colors should belong to the same
class; as for the odd case, without loss of generality, let us assume that $1$ and $2$ both belong to Class 2.

\begin{observation}\label{zero_term_obs_2}
 If $\mu(G)\geq 3$, then $$ w(G_{A},\mathbf{2}_A)
w(G_{BU},\mathbf{2}_B\mathbf{1}_U) w(G'_{AU},\mathbf{2}_A\mathbf{1}_U) w(G_{B},\mathbf{2}_B)\neq0$$
In other words, the coloured square of $c=\mathbf{2}_A\mathbf{2}_B\mathbf{1}_U$ is solid. 
%$w(\mathbf{2}_A)w(\mathbf{2}_B\mathbf{1}_U)w(\mathbf{2}_A\mathbf{1}_{U})w(\mathbf{2}_B)\neq 0$   
\end{observation} 
\begin{proof}

Let $c_i = i_Ai_Bi_U$. For $i=1,2$, by assumption  ${\cal V}_{c_i} \ne 0$, and therefore it follows that ${\cal V}_{c_1}^l, {\cal V}_{c_1}^r,{\cal V}_{c_2}^l,{\cal V}_{c_2}^r \neq 0$.
From this, we can infer that $0 \ne {\cal V}_{c_1}^l.{\cal V}_{c_2}^r = w'(1_A1_U).w(2_B) = {\cal V}_{c'}$, for the non-monochromatic coloring $c'= 1_A2_B1_U$. 
By a similar argument, $0 \ne    {\cal V}_{c_2}^l.{\cal V}_{c_1}^r = w'(2_A2_U).w(1_B) = {\cal V}_{c''}$  for the non-monochromatic coloring $c'' =2_A1_B2_U$.
From observation \ref {even_solid_square},  we infer that both square$(c')$ and square$(c'')$ are solid.  It follows that 
$H_{c'}^t = w(2_B1_U) \ne 0$ and $H_{c''}^b = w(2_A) \ne 0 $.  Now  $w(2_B1_U).w(2_A) = H_{\delta}$, where $\delta = 2_A2_B1_U$.  It follows that square$(\delta)$ is
solid. %, contradicting observation  \ref {evencancelling}.
\end{proof}
We will now prove that $\delta$ is fragile and thus contradicting \cref{zero_term_obs_2}.
%As \cref{evencancelling} and \cref{zero_term_obs_2} contradict each other, our assumption that $\mu(G) \geq 3$ is wrong. 

\begin{observation}
\label{evencancelling}
$$ w(G_{A},\mathbf{2}_A)
w(G_{BU},\mathbf{2}_B\mathbf{1}_U)
w(G'_{AU},\mathbf{2}_A\mathbf{1}_U)
w(G_{B},\mathbf{2}_B)=0$$
In other words, the coloured square of $c=\mathbf{2}_A\mathbf{2}_B\mathbf{1}_U$ is fragile. 
\end{observation}

\begin{proof}

Let $c_1= \mathbf{2}_A\mathbf{1}_B\mathbf{1}_{U}$,  $c_2= \mathbf{1}_A\mathbf{2}_B\mathbf{1}_{U}$ and $c_3= \mathbf{2}_A\mathbf{2}_B\mathbf{1}_{U}$.  By observation \ref {even_solid_square}, it is easy to see that
$   H_{c_1}.H_{c_2}.{\cal V}_{c_3}= - {\cal V}_{c_1}.{\cal V}_{c_2}. H_{c_3}$.  Expanding this we get,

$$H_{c_1}^b.H_{c_1}^t. H_{c_2}^b.H_{c_2}^t.{\cal V}_{c_3}^l.{\cal V}_{c_3}^r = - {\cal V}_{c_1}^l. {\cal V}_{c_1}^r. {\cal V}_{c_2}^l.{\cal V}_{c_2}^r.H_{c_3}^b.H_{c_3}^t$$

Now, substituting the actual weights for each term, 
\begin{multline*}
    w(\mathbf{2}_A)w(\mathbf{1}_B\mathbf{1}_U) w(\mathbf{1}_A)w(\mathbf{2}_B\mathbf{1}_U) w(\mathbf{2}_A\mathbf{1}_U)w(\mathbf{2}_B)\\=-w'(\mathbf{2}_A\mathbf{1}_U)w(\mathbf{1}_B)w'(\mathbf{1}_A\mathbf{1}_U)w(\mathbf{2}_B)w(\mathbf{2}_A)w(\mathbf{2}_B\mathbf{1}_U)
\end{multline*}
It follows that 
\begin{multline*}
w(G_{A},\mathbf{2}_A)
w(G_{BU},\mathbf{2}_B\mathbf{1}_U)
w(G'_{AU},\mathbf{2}_A\mathbf{1}_U)
w(G_{B},\mathbf{2}_B) \\
\times (w(G_{A},\mathbf{1}_A)w(G_{BU},\mathbf{1}_B\mathbf{1}_U)+w(G'_{AU},\mathbf{1}_A\mathbf{1}_U)w(G_{B},\mathbf{1}_B))=0
\end{multline*}
Since $\mathbf{1}_A\mathbf{1}_B\mathbf{1}_U$ is a monochromatic colouring, it must have a weight $1$. Therefore,
\begin{equation*}
    0\neq 1= w(G,\mathbf{1}_A\mathbf{1}_B\mathbf{1}_U)=w(G_{A},\mathbf{1}_A)w(G_{BU},\mathbf{1}_B\mathbf{1}_U)+w(G'_{AU},\mathbf{1}_A\mathbf{1}_U)w(G_{B},\mathbf{1}_B)
\end{equation*}
It now follows that, $$w(G_{A},\mathbf{2}_A)
w(G_{BU},\mathbf{2}_B\mathbf{1}_U)
w(G'_{AU},\mathbf{2}_A\mathbf{1}_U)
w(G_{B},\mathbf{2}_B) =0$$
Therefore, the coloured square of the colouring $\mathbf{2}_A\mathbf{2}_B\mathbf{1}_U$ is fragile. 
\end{proof}
Hence, from \cref{zero_term_obs_2} and \cref{evencancelling} , if $|A|$ is even, then ${\mu}(G)\leq 2$.
\end{proof}

Therefore, from \cref{conn0obs}, \cref{conn1obs}, \cref{oddlemma} and \cref{evenlemma}, it follows that, for a graph $G$, if $\kappa(G)\leq 2$, then $\mu(G)\leq 2$. 
\end{proof}

\section{Reduction}
\label{reduction}
We prove a stronger theorem than \cref{reduction_theorem_main} as stated below.
\begin{theorem} 
\label{specific_theorem}
Let $G$ be a multi-graph with a vertex cut $S$ of size $3$. Let $V_1$ and $V_2$ be a partition of $V(G)\setminus S$ such that $V_1$ and $V_2$ are non-empty and there are no edges between $V_1$ and $V_2$ in $G$. Moreover, let $|V_1|$ be odd and $|V_2|$ be even. There exists a graph $G'$ such that $|V(G')|\leq |V_1|+3 \leq |V(G)|-2$ and $\mu(G')\geq \mu(G)$. 
\end{theorem}
%\textcolour{red}{ Let $V_1$ and $V_2$ be a partition of $V(G)\setminus S$ such that there are no edges between $V_1$ and $V_2$ in $G$ and both are non empty.}\\
Let $w,c$ be a colouring and weight assignment of $G$ for which $\mu(G,c,w)=\mu(G)$. 
Then $G'$ would be a graph on the vertex set $V_1 \bigsqcup S$. The edge set, the edge weight function
and the edge colouring of $G'$ would be the same as in $G[V_1 \bigsqcup S]$ except for the edges with  both endpoints
inside $S$; 
We redefine the set of edges, edge-weight function, edge-colouring 
within $S$. Let $c'$ and $w'$ represent the edge-colouring and the edge-weight function of $G'$. We will
show that $\mu(G')\geq \mu(G',c',w')=\mu(G,c,w)=\mu(G)$. 

Let $V=V(G)$ and $S=\{u_1,u_2,u_3\}$. Note that any perfect matching of $G$,
should match an odd number of vertices of $S$ to $V_1$ (since $|S|$ is odd).
Therefore  the perfect matchings of  $G$ can be grouped into $4$ types:

\textbf{Type 0:}   Let ${\cal P}_0$ denote the set of all perfect matchings on $G[V_1\cup S]$ in which 
all the three vertices of  $S$  are matched with vertices in $V_1$. Let $W_0(vc)$ denote the
sum of weights of the perfect matchings from ${\cal P}_0$ that induce the vertex colouring $vc$ on ${V_1 \bigsqcup S}$.
Let ${\cal P}_0'$ denote the set of all perfect matchings on $G[V_2]$. Let $W_0'(vc)$  denote
the sum of weights of the perfect matchings from ${\cal P}_0'$  that induce the vertex colouring $vc$ on ${V_2}$. 
Type 0 perfect matchings of $G$ are the perfect matchings that belong to ${\cal P}_0 \times {\cal P}_0'$.
Clearly, the sum of the weights of Type 0 perfect matchings that induce the colouring $vc$ equals $W_0(vc)W_0'(vc)$.

\textbf{Type $i$:}
For $i \in \{1,2,3\}$, Let ${\cal P}_i$ denote the set of
perfect matchings of $V_1 \bigsqcup \{u_i\}$, and let ${\cal P}_i'$ denote the
set of perfect matchings of  $V_2 \bigsqcup (S \setminus \{u_i\}) $. Let $W_i(vc)$ denote the sum of weights of all perfect matchings from ${\cal P}_i$  that induce  the vertex colouring $vc$ on ${V_1 \bigsqcup \{u_i\}}$ and $W_i'(vc)$ denote the sum of weights of all perfect matchings from ${\cal P}_i'$ that induce the colouring  $vc$ on ${V_2 \bigsqcup (S \setminus \{u_i\})}$. Type $i$ matchings of $G$ are the perfect matchings that belong to 
${\cal P}_i \times {\cal P}_i'$.   Clearly, the sum of weights of Type $i$ perfect matchings that 
induce the colouring $vc$  equals $W_i(vc) W_i'(vc).$ From the above discussion,  it is easy to see that
\begin{equation}
\label{4termsum}
w(vc)=\sum_{i \in \{0,1,2,3\}} W_i(vc)W_i'(vc)    
\end{equation}
Recall that $c_{V_2}$ denotes the monochromatic vertex colouring with the colour $c$ on $V_2$. Let us  partition $[\mu(G)]$  into $\mathcal{C}_1\sqcup \mathcal{C}_2$ such that $c \in \mathcal{C}_1$, if and only if there exists some colouring $c'$ on $V_1 \sqcup S$ such that $W'_0(c_{V_2})W_0(c') \neq 0$. The remaining colors from $[\mu(G)]$ belong to $\mathcal{C}_2$ (Note that {this happens if $W'_0(c_{V_2}) = 0$} or if for all colourings $c'$ on $V_1 \sqcup S$, $W_0(c') = 0$.).  We will now prove \cref{specific_theorem} in two cases. When $\mathcal{C}_1 = \emptyset$ (Theorem \ref{specific_theorem_easy_case}) and $\mathcal{C}_1 \neq \emptyset$ (Theorem \ref{specific_theorem_hard_case})

\begin{subtheorem}
\label{specific_theorem_easy_case}
Let $G$ be a multi-graph with a vertex cut of size $3$. If $\mathcal{C}_1 = \emptyset$ (as defined earlier), then $\mu(G) \leq \mu(K_4)$
\end{subtheorem}

\begin{subtheorem}
\label{specific_theorem_hard_case}
Let $G$ be a multi-graph with a vertex cut of size $3$. If $\mathcal{C}_1 \neq \emptyset$ (as defined earlier), then there exists a graph $G'$ such that $|V(G')|\leq |V_1|+3 \leq |V(G)|-2$ and $\mu(G')\geq \mu(G)$. Recall that $V_1$ and $V_2$ is a partition of $V(G)\setminus S$ such that $V_1$ and $V_2$ are non-empty and there are no edges between them. Moreover $|V_1|$ is odd and $|V_2|$ is even. 
\end{subtheorem}

\subsection{Construction for Theorem \ref{specific_theorem_easy_case}}
\label{construction_easy_case}
Let $V(G')=\{v_0,v_1,v_2,v_3\}$. The reader may mentally map the vertices $v_1,v_2,v_3$ to the vertices $u_1,u_2,u_3$ of the vertex cut $S$ and $v_0$ to the set of vertices $V_1$. Note that $G'$ is a multi-graph (without self loops) and each pair of vertices from $V(G')$ may have many edges between them. In fact, we would define $\mu(G)^2$ number of edges between each pair of vertices in $\{v_0,v_1,v_2,v_2\}$, one edge for each ordered pair in $[\mu(G)]\times[\mu(G)]$. 

Let $i,j \in \{1,2,3\}$ with $i \neq j$. Let $(p,q) \in [\mu(G)]\times[\mu(G)]$. We will define an edge $e$ for each such $(p,q)$ between $v_i$ and $v_j$. This edge $e$ would be coloured such that the $v_i$-half edge of $e$ (i.e., $e_{v_i}$) has colour $p$ and  the $v_j$-half edge of $e$ (i.e., $e_{v_j}$) has colour $q$. The weight of the edge would be 
%\begin{equation}\label{eq_v23}
$w(e)= \sum_{c \in \mathcal{C}_2} w(c_{V_2}p_{u_i}q_{u_j})$
%\end{equation}
The reader may recall that $c_{V_2}p_{u_i}q_{u_j}$ represents the colouring in which $V_2,u_i,u_j$ are coloured with $c,p,q$ respectively. Its weight is the weight of the induced subgraph of $G$ on $V_2 \cup \{u_i,u_j\}$ filtered out by the colouring $c_{V_2}p_{u_i}q_{u_j}$. In this case, $\mathcal{C}_1 = \emptyset$; so $\mathcal{C}_2 = [\mu(G)]$; and therefore there are $\mu(G)$ terms in the summation. 

%Note that a typical term 
We will now consider the edges between $v_0$ and $v_i$, when $i \in \{1,2,3\}$. Let $(p,q) \in [\mu(G)]\times[\mu(G)]$. We will define an edge $e$ for each such $(p,q)$ between $v_0$ and $v_i$. This edge $e$ would be coloured such that the $v_0$-half edge of $e$ (i.e., $e_{v_0}$) has colour $p$ and  the $v_i$-half edge of $e$ (i.e., $e_{v_i}$) has colour $q$. The weight of such an edge $e$ is defined as 
%\begin{equation}\label{eq_v1}
$w(e)= w(p_{V_1}q_{u_i})$
%\end{equation}

\subsection{Proof of construction for Theorem \ref{specific_theorem_easy_case}}
\label{proof_for_construction_easy_case}
\label{k0proof}

Consider any vertex colouring $vc:V'\to \mathbf{N}$ on $G'$. We will prove that $w(V(G'),vc)=0$, if $vc$ is non-monochromatic and $w(V(G'),vc)=1$,  if $vc$ is monochromatic.

Let the vertex colouring $vc$ be $i_{v_0}j_{v_1}k_{v_2}l_{v_3}$. To find the weight of $vc$, we consider the subgraph of $G'$ filtered out by $vc$. For instance, between the vertices $v_2$ and $v_3$, the colouring $vc$ would filter exactly one edge, and such an edge would have the $v_2$-half edge of colour $k$ and $v_3$-half edge of colour $l$. Clearly, this gives a graph which is isomorphic to $K_4$ (with some edges possibly getting a weight of zero). Observe that there are only three perfect matchings in $K_4$. So, it is now easy to find the weight of $vc$ by enumeration.

As an example, consider the vertex colouring in which $v_0,v_1$ are coloured red and $v_2,v_3$ are coloured green and blue, respectively. Its filtering is shown in \cref{fig:filtering_k4}. Its weight would be $w_1w_1'+w_2w_2'+w_3w_3'$.

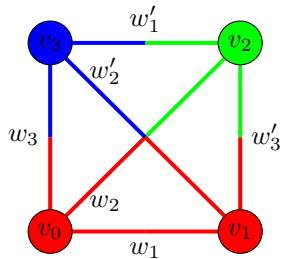
\begin{figure}[htbp]
\centering
\begin{tikzpicture}[scale=2.5]

% Define vertices
\coordinate (A) at (0,0);
\coordinate (B) at (1,0);
\coordinate (C) at (1,1);
\coordinate (D) at (0,1);

% Define colours
\definecolor{OuterEdgecolour}{RGB}{255,0,0} % Red
\definecolor{HalfcolorA}{RGB}{0,128,0} % Green
\definecolor{HalfcolorB}{RGB}{255,0,0} % Red

% Draw outer edges with one colour each
%\draw[OuterEdgecolour, line width=1.5pt, red] (A) -- (B);

\draw[line width=1.5pt, red] (A) -- ($(A)!0.5!(B)$) node[right, below] {\textcolor{black}{$w_1$}};
\draw[line width=1.5pt, red] ($(A)!0.5!(B)$) -- (B);

\draw[HalfcolorA, line width=1.5pt, red] (B) -- ($(B)!0.5!(C)$) node[right, right] {\textcolor{black}{$w'_3$}};
\draw[HalfcolorB, line width=1.5pt, green] ($(B)!0.5!(C)$) -- (C);

\draw[HalfcolorA, line width=1.5pt, green] (C) -- ($(C)!0.5!(D)$) node[right, above] {\textcolor{black}{$w'_1$}};
\draw[HalfcolorB, line width=1.5pt, blue] ($(C)!0.5!(D)$) -- (D) ;

\draw[HalfcolorA, line width=1.5pt, blue] (D) -- ($(D)!0.5!(A)$) node[below, left] {\textcolor{black}{$w_3$}};
\draw[HalfcolorB, line width=1.5pt, red] ($(D)!0.5!(A)$) -- (A);

% Draw inner edges with half-edge colouring
\draw[HalfcolorA, line width=1.5pt, red] (A) -- ($(A)!0.5!(C)$) node[pos=0.3, right] {\textcolor{black}{$w_2$}};
\draw[HalfcolorB, line width=1.5pt, green] ($(A)!0.5!(C)$) -- (C);

\draw[HalfcolorA, line width=1.5pt, red] (B) -- ($(B)!0.5!(D)$);
\draw[HalfcolorB, line width=1.5pt, blue] ($(B)!0.5!(D)$) -- (D) node[pos=0.7, right] {\textcolor{black}{$w'_2$}};

% Mark vertices with labels
 \fill (A) circle (2pt) node [draw, circle, fill=red, inner sep=2pt] {$v_0$};
 \fill (B) circle (2pt) node [draw, circle, fill=red, inner sep=2pt] {$v_1$};
 \fill (C) circle (2pt) node [draw, circle, fill=green, inner sep=2pt] {$v_2$};
  \fill (D) circle (2pt) node [draw, circle, fill=blue, inner sep=2pt] {$v_3$};

\end{tikzpicture}

\caption{$K_4$ obtained by a filtering operation} %in which $v_0,v_1$ are coloured red and $v_2,v_3$ are coloured green and blue, respectively}
\label{fig:filtering_k4}
\end{figure}

We will first compute the weight of perfect matching $\{\{v_0,v_1\}, \{v_2,v_3\}\}$ in $vc$. By substituting the edge weights from the construction%(Equations \cref{eq_v23,eq_v1})
, we get the weight of the perfect matching to be
$$ = w(i_{V_1}j_{u_1}) \sum_{c \in [\mu(G)]} w(c_{V_2}k_{u_2}l_{u_3})= \sum_{c \in [\mu(G)]} w(i_{V_1}j_{u_1}) w(c_{V_2}k_{u_2}l_{u_3})$$
%Rearranging the terms in the summation, we get $$= \sum_{c \in [\mu(G)]} w(i_{V_1}j_{u_1}) w(c_{V_2}k_{u_2}l_{u_3})$$
As $w(i_{V_1}j_{u_1}) w(c_{V_2}k_{u_2}l_{u_3})$ is exactly the sum of weights of Type $1$ perfect matchings for the vertex colouring $i_{V_1}j_{u_1}k_{u_2}l_{u_3}c_{V_2}$ %(from Equation \cref{4termsum})
\begin{equation}
\label{type1sum}
= \sum_{c \in [\mu(G)]} W_1(i_{V_1}j_{u_1}k_{u_2}l_{u_3}c_{V_2})W'_1(i_{V_1}j_{u_1}k_{u_2}l_{u_3}c_{V_2})
\end{equation}
Similarly, the weight of the weight of perfect matching $\{\{v_0,v_2\}, \{v_1,v_3\}\}$ in $vc$ is \begin{equation}
\label{type2sum}
= \sum_{c \in [\mu(G)]} W_2(i_{V_1}j_{u_1}k_{u_2}l_{u_3}c_{V_2})W'_2(i_{V_1}j_{u_1}k_{u_2}l_{u_3}c_{V_2})
\end{equation}
and the weight of the weight of perfect matching $\{\{v_0,v_3\}, \{v_1,v_2\}\}$ in $vc$ is 
\begin{equation}
\label{type3sum}
= \sum_{c \in [\mu(G)]} W_3(i_{V_1}j_{u_1}k_{u_2}l_{u_3}c_{V_2})W'_3(i_{V_1}j_{u_1}k_{u_2}l_{u_3}c_{V_2})
\end{equation}
As the weight of the vertex colouring $vc'$ is the sum of these three perfect matchings, by adding the equations \cref{type1sum,type2sum,type3sum}, we get 
$$w(vc')= \sum_{r \in \{1,2,3\}} \sum_{c \in [\mu(G)]} W_r(i_{V_1}j_{u_1}k_{u_2}l_{u_3}c_{V_2})W'_r(i_{V_1}j_{u_1}k_{u_2}l_{u_3}c_{V_2})$$
By rearranging the summation and using \cref{4termsum}, 
$$= \sum_{c \in [\mu(G)]} \sum_{r \in \{1,2,3\}} W_r(i_{V_1}j_{u_1}k_{u_2}l_{u_3}c_{V_2})W'_r(i_{V_1}j_{u_1}k_{u_2}l_{u_3}c_{V_2}) = \sum_{c \in [\mu(G)]} w(i_{V_1}j_{u_1}k_{u_2}l_{u_3}c_{V_2})$$

First suppose that $vc=i_{v_0}j_{v_1}k_{v_2}l_{v_3}$ is non-monochromatic. Clearly, $i_{V_1}j_{u_1}k_{u_2}l_{u_3}c_{V_2}$ is also non-monochromatic for all $c \in [\mu(G)]$. Since the weight of all such colourings is zero, their sum is also zero. It now follows that $w(vc)=0$

On the other hand, suppose $vc=i_{v_0}j_{v_1}k_{v_2}l_{v_3}$ is monochromatic, i.e., $i=j=k=l$. Clearly, if $c=i$, $i_{V_1}j_{u_1}k_{u_2}l_{u_3}c_{V_2}$ is also monochromatic and has weight $1$. If $c\neq i$, $i_{V_1}j_{u_1}k_{u_2}l_{u_3}c_{V_2}$ is non-monochromatic and has weight $0$. Since we take the sum over all $c$, it now follows that $w(vc)=1$.

\subsection{Applications of construction for Theorem \ref{specific_theorem_easy_case}}
\label{red_appli}
%\textcolour{red}{$\mu$ is defined for skeletons which are simple. It is incorrect to write minimum degree of multi-graph $G$}
\begin{corollary}
\label{coll1}
If the minimum degree of a graph $G$ is $3$, then $\mu(G)\leq 3$ 
\end{corollary}
\begin{proof}
Let $u$ be a three-degree vertex, and $x,y,z$ be its neighbours. Let $V_1=\{u\}$, $S=\{x,y,z\}$ and $V_2=V\setminus \{u,x,y,z\}$. Note that there are no Type $0$ matchings as $u$ can match with at most one vertex in $S$ in any perfect matching. Therefore, $W_{0}(c)$ is zero for all $c \in [\mu(G)]$. It is now easy to see that $\mathcal{C}_1=\emptyset$. Therefore, from Theorem \ref{specific_theorem_easy_case}, $\mu(G) \leq \mu(K_4)$. From \cref{kevin}, it is known that the matching index of graphs with $4$ vertices is at most $3$. Therefore,  $\mu(G)\leq \mu(K_4) = 3$.
\end{proof}
%\textcolour{red}{cite \cite{Kevin} above? proved with the use of computers seems unnecessary here?}
\begin{corollary}
\cref{krenn_conjecture} is true for all graphs whose maximum degree is at most $3$.
\end{corollary}
\begin{proof}
From \cref{coll1}, we know that $\mu(G)\leq 3$. Towards a contradiction, let $\mu(G)=3$ for graph with $|V(G)|>$4 and maximum degree $3$. Therefore, there exists a colouring $c$ and weight assignment $w$ such that $\mu(G,c,w)=3$. Let the three colours be $1,2,3$.

We first claim that between any pair of vertices, there is at most one non-zero edge incident on it. Suppose not. Then, there exist vertices $u$ and $x_1$ with multiple non-zero edge between them. Since the maximum degree of the skeleton $G$ is at most $3$, there exists a vertex set (for instance, all neighbours of $u$ if the degree is $3$) $\{x_1,x_2,x_3\}$ which separates $u$ from the $V-\{u,x_1,x_2,x_3\}$. Let $V_1=\{u\}$, $S=\{x_1,x_2,x_3\}$ and $V_2=V\setminus \{u,x_1,x_2,x_3\}$. Recall that $\mathcal{C}_1 = \emptyset$. By the construction from \cref{construction_easy_case}, the weights of edges between $u$ and the vertices $\{x_1,x_2,x_3\}$ remain unchanged. Therefore, we obtain a graph with $4$ vertices of dimension $3$ such that a pair of vertices have multiple non-zero edges between them. But this is not possible from \cref{kevin}. Therefore, there are no multi-edges in $G_c^w$.

Since the matching index is $3$, there are perfect matchings (of non-zero edges) of colours $1,2,3$. Therefore, from \cref{bogdanov}, there must be at least one non-monochromatic perfect matching $M$ (of non-zero edges), say inducing the non-monochromatic 
vertex colouring $vc$. But there is exactly one non-zero edge of colours $1,2,3$ incident on $u$. Therefore, for any vertex colouring $vc$, there can be at most one perfect matching $M'$ inducing $vc$. It now follows that $w'(M)=w'(vc)=0$. Therefore, there must be an edge, say of colour $1$, whose weight is zero; hence, the monochromatic vertex colouring of $1$ must be zero. Contradiction. 

Therefore, $\mu(G)\leq 2$ when $|V(G)|>4$.
\end{proof}

\subsection{{Construction for Theorem \ref{specific_theorem_hard_case}}} %The case when ${\cal C}_1 \ne \emptyset$ }}
\label{construction_1}
%\label {wt_of_edge_in_harder_case}
Recall that $S=\{u_1,u_2,u_3\}$ is  a vertex cut separating $V_1$ from  $V_2$ in the graph $G$, where $|V_1|$  is odd, $|V_2|$ is even. %and $G[V_i], 1 \le i \le 2,$ consists of a collection of connected components of $G \setminus S$. 
Assume that ${\cal C}_1  \ne \emptyset$. For this case,  we will construct an edge-weighted, edge coloured multi-graph  $G'$ with $V(G')=V_1 \cup S$ and $\mu(G') \ge \mu(G)$. Since $|V_2| \ge 2$, $|V(G')| \leq |V(G)|-2$.

 If a pair of vertices is such that at least one of them lies in $V_1$, then the set of edges in $G'$ between this pair of vertices is the same as those in $G$ with the same weights and colours. Between the pairs of vertices with both vertices from $S$, we define one edge for each pair of colours in $[\mu(G)]\times[\mu(G)]$.  Thus there would be $(\mu(G))^2$ edges between each pair,
 $\{u_i, u_j\}, i \neq j$.

We now describe how to assign weight to a coloured edge $e$ between the vertices $u_i$ and $u_j$, where $i < j$ and $i,j \in \{1,2,3\}$ such that the  $u_i$-half  of $e$ is coloured $p$ and the $u_j$-half  of $e$ is coloured $q$. 
%Let $vc'$ be the vertex colouring induced by edge $e$ on its vertices that is $vc'(u_i)=p$ and %$vc'(u_j)=q$. 

 \begin {eqnarray}
 \label {wt_of_edge_in_harder_case}
w'(e)=  \sum_{ c \in \mathcal{C}_2}  W(c_{V_2}p_{u_i}q_{u_j})          + \dfrac{1}{|\mathcal{C}_1|} \sum_{c \in \mathcal{C}_1} \dfrac{W (  c_{V_2}p_{u_i}q_{u_j}    )}{W(c_{V_2})}
\end {eqnarray}

\begin {remark}
  Recall that our plan is to remove $V_2$ and the
edges incident on $V_2$ completely in order to  get the reduced  graph $G'$. This
should be done without losing the information about 
the weights of monochromatic vertex colourings of $V_2$ to make sure that $\mu(G')$ does not become
smaller than $\mu(G)$. 
The weight assignment is  similar in spirit to the weight assignment done for the construction of Theorem \ref {specific_theorem_easy_case}. 
The expression here is more complicated in this case,  because of the adjustments required to make it work: This will be clear when the reader goes
through the proof carefully.
\end {remark}

\subsection{Proof of the construction for Theorem \ref{specific_theorem_hard_case}}
\label{construction_1_proof}

Consider any vertex colouring $vc':V_1 \cup S \to \mathbf{N}$ of $G'$. Let us denote $vc'$ more
explicitly, using the notation describe in \cref {construction_easy_case},  as $  (\alpha)_{V_1} i_{u_1} j_{u_2} k_{u_3}$.  Note that $\alpha$ here is 
the vertex colouring induced on $V_1$ by the vertex colouring $vc'$, $i,j,k$ are the colours of $u_1,u_2,u_3$ respectively
under the vertex colouring $vc'$.   (Note that we use the notation  $(\alpha)_{V_1}$ to emphasize that this is not necessarily a
monochromatic colouring of $V_1$, using a single colour named  $\alpha$; rather it can be any vertex colouring, monochromatic  or
non-monochromatic.) 
We will use the notation of $w'$ to denote the weights of vertex colourings of  $G'$ and
its subgraphs and $w$ to denote the weights of vertex colourings of  $G$ and its
subgraphs.  For example,  $w(i_{V_1}j_{S})$ is the weight of the vertex colouring $i_{V_1}j_S$ with respect to the edge-set of $G$, whereas $w'(i_{V_1}j_S)$ denotes the weight of the same vertex colouring with respect to the edge-set of $G'$. 
Our intention is to prove that $w'((\alpha)_{V_1} i_{u_1} j_{u_2} k_{u_3} )=0$, whenever $ (\alpha)_{V_1} i_{u_1} j_{u_2} k_{u_3} $ is non-monochromatic and $w'( (\alpha)_{V_1} i_{u_1} j_{u_2} k_{u_3} )\neq 0$,  whenever  $ (\alpha)_{V_1} i_{u_1} j_{u_2} k_{u_3}$ is monochromatic, i.e. $i=j=k$ and
$\alpha$ is a monochromatic vertex colouring of $V_1$ using colour $i$,  i.e. $(\alpha)_{V_1} = i_{V_1}$.  The reader may note that we are not insisting the weight to be equal to $1$ in the monochromatic case; non-zero is sufficient since
we can then use the Scaling Lemma (\cref{scaling_generalization}) to scale the weights of the monochromatic vertex colourings to $1$,  thus constructing an edge-weight function $w'$ and edge-colouring $c'$ of $G'$ such that  $\mu(G',c',w')=\mu(G)$ implying $\mu (G') \ge \mu (G)$.

Recall that $vc' = (\alpha)_{V_1} i_{u_1} j_{u_2} k_{u_3}$ filters out a simple graph from $G'$ and the weight of $vc'$ is the sum of
weights of all the perfect matchings (PMs)  in this filtered out simple graph.  Perfect matchings (PMs)  of
this graph can be grouped into 4 categories: 
(1) Type $0'$: PMs containing none of the edges with both endpoints in $S$.   
(2) Type $1'$:  PMs containing the edge
$(u_2, u_3)$  (3) Type $2'$:   PMs containing the edge $(u_1, u_3)$  (4)  Type $3'$: PMs containing
the edge $(u_1,u_2)$  
For $t=0,1,2,3$,  we denote the total weight of Type $t'$  PMs
by $W'_t$.   Clearly $w'(vc') = \sum_{0 \le t \le 3} W_t'$.

Now let us consider a  corresponding vertex colouring of $G$,
$vc = (\alpha)_{V_1} i_{u_1} j_{u_2} k_{u_3} c_{V_2} $, which is
obtained by taking the same vertex colouring $vc'$ for $V_1 \cup S$, 
and  then extending it by the  monochromatic vertex colouring $c_{V_2}$ of $V_2$, using the colour $c$.  Since we may use
any colour $c \in [\mu(G)]$ to extend the vertex colouring 
$vc'$ of $V_1 \cup S$ to a vertex colouring of $V_1 \cup S \cup V_2$, it is more appropriate to call the extended colouring  
$(\alpha)_{V_1}i_{u_1}j_{u_2}k_{u_3}c_{V_2}$, it is  better to denote  
$vc(c)$, rather than  just $vc$.

The weight of $vc$ on $G$ can 
also be decomposed into $4$ terms corresponding to $4$ different
groups of perfect matchings of  the subgraph filtered out by $vc$
from $G$.  

(1) Type $0$:  The PMs in which all the three vertices of
$S$ are matched to vertices in $V_1$. 

(2) Type  $t$ for $t=1,2,3$:  The PMs in which only $u_t$  is
matched to some vertex of $V_1$, and the remaining two vertices of  $S$ are either matched to each other or to  vertices of
$V_2$. 

The total weight of the perfect matchings (in the
the subgraph of $G$ filtered out by the vertex colouring $vc(c)$) of  Type$t, 0 \le t \le 3$ will
be denoted by $W_t(c)$ where $c \in [u(G)]$.  Clearly $w(vc(c)) =
\sum_{0 \le t \le 3} W_t(c)$. Now we  show that $W_t'$  can be expressed  as a (weighted)  sum of  $W_t(c)$ over the colours $c \in  [\mu(G)]$.

\begin {observation}
 For $t=1,2,3:  W_t' = \sum_{c \in {\cal C}_2} W_t(c) + \dfrac {1} {|{\cal C}_1|} \sum_{c \in {\cal C}_1} \dfrac {W_t(c)} {w(c_{V_2})} $  
 \end {observation}

\begin{proof}
Let us calculate the total weight $W_1'$ of Type 1 perfect matchings of the vertex colouring
$\alpha_{V_1} i_{u_1} j_{v_2} k_{v_3}$.  (The case when $t=2,3$ is
similar.) Clearly these PMs
are obtained by adding the edge $(u_2, u_3)$  of colour $(j,k)$ to each  perfect matching of the induced subgraph 
on $V_1 \cup \{u_1\}$ (after filtering out by the vertex colouring $\alpha_{V_1} i_{u_1}$). 
Therefore the total weight of these PMs can be written as
$W_1'  = w'(\alpha_{V_1} i_{u_1} ) w' (e)$ 
where $e$ is the edge between $u_2$ and $u_3$  of colour $(i,j)$,
that is, the edge $e$ with $u_2$-half 
of $e$ coloured  $j$ and $u_3$-half of $e$ coloured $k$.
Now substituting  for $w'(e)$, the right hand side of equation
\ref {wt_of_edge_in_harder_case},  we get
\begin {eqnarray}
W_1' =  w'(\alpha_{V_1}i_{u_1})  \left ( 
\sum_{ c \in \mathcal{C}_2}  w(c_{V_2}j_{u_2}k_{u_3})          + \dfrac{1}{|\mathcal{C}_1|} \sum_{c \in \mathcal{C}_1} \dfrac{w(  c_{V_2}j_{u_2}k_{u_3}    )}{w(c_{V_2})}  \right ) 
\end {eqnarray}
Note that $w'((\alpha)_{V_1}i_{u_1}) = w((\alpha)_{V_1}i_{u_1})$
 since in the induced subgraph on $V_1 \cup \{u_1\}$ the edge set, weights and colour are same for  both $G$ and $G'$. It follows that, 

\begin {eqnarray}
W_1' =    
\sum_{ c \in \mathcal{C}_2} 
w((\alpha)_{V_1}i_{u_1}))
w(c_{V_2}j_{u_2}k_{u_3})          + \dfrac{1}{|\mathcal{C}_1|} \sum_{c \in \mathcal{C}_1} \dfrac{   w((\alpha)_{V_1}i_{u_1}))   w (  c_{V_2}j_{u_2}k_{u_3}    )}{w(c_{V_2})}   
\end {eqnarray}

Noting that $w(\alpha_{V_1}i_{u_1})
w(c_{V_2}j_{u_2}k_{u_3}) =  w( \alpha_{V_1}i_{u_1}
   j_{u_2}k_{u_3} c_{V_2}) = W_1(c)$  we get,

$$
W_1' =  \sum_{ c \in \mathcal{C}_2} 
w(\alpha_{V_1}i_{u_1}
   c_{V_2}j_{u_2}k_{u_3})          + \dfrac{1}{|\mathcal{C}_1|} \sum_{c \in \mathcal{C}_1} \dfrac{   w(\alpha_{V_1}i_{u_1}    c_{V_2}j_{u_2}k_{u_3}  )}{w(c_{V_2})} =
   \sum_{ c \in \mathcal{C}_2} 
W_1(c)           + \dfrac{1}{|\mathcal{C}_1|} \sum_{c \in \mathcal{C}_1} \dfrac{ W_1(c) } {w(c_{V_2})}  
$$

Similar arguments allow us establish the required result for
$t=2,3$ also. 
\end{proof}

\begin{observation}
\label{sum_eq_fin}
$w(vc') =  \sum_{ c \in \mathcal{C}_2} 
w(vc(c))  +  \dfrac{1}{|\mathcal{C}_1|} \sum_{c \in \mathcal{C}_1} \dfrac{w(vc(c)) } {w(c_{V_2})}$
\end{observation}
\begin{proof}
Now $w'(vc') = W_0' + W_1' + W_2' + W_3'$ 

  =  $W_0' +   \sum_{ c \in \mathcal{C}_2} 
\left ( W_1(c)  + W_2(c) + W_3(c) \right )         + \dfrac{1}{|\mathcal{C}_1|} \sum_{c \in \mathcal{C}_1} \dfrac{ W_1(c) + W_2(c) + W_3(c) } {w(c_{V_2})} $

Recall that for any colour $c \in [\mu(G)]$, $w(vc(c)) = 
W_0(c) + W_1(c) + W_2(c) + W_3(c)$.
Note that for colours $c \in {\cal C}_2$, $W_0(c) = W_0'.w(c_{V_2}) $; and therefore  $W_1(c) + W_2(c) +
W_3(c) = w(vc(c))$, the weight of the vertex colouring $\alpha_{V_1}
i_{u_1} j_{u_2} k_{u_3} c_{V_2}$. 
Now $W_0'$ can be trivially rewritten as $\dfrac {1} {|{\cal C}_1|} \sum_{c \in {\cal C}_1 } \dfrac {W_0'.w(c_{V_2}) } {w(c_{V_2})}$.
Since $W_0'.w(c_{V_2}) = W_0(c)$, this expression can be 
rewritten as $\dfrac {1} {|{\cal C}_1|} \sum_{c \in {\cal C}_1 } \dfrac {W_0(c) } {w(c_{V_2})}$.

So we can combine the terms of this expression, term by term
with the terms inside the sum over $c \in {\cal C}_1$ and rewrite the expression 
as follows:
\begin {eqnarray}
=\sum_{ c \in \mathcal{C}_2} 
w(vc(c))  +  \dfrac{1}{|\mathcal{C}_1|} \sum_{c \in \mathcal{C}_1} \dfrac{W_0(c) + W_1(c) + W_2(c) + W_3(c) } {w(c_{V_2})} 
%= \sum_{ c \in \mathcal{C}_2} 
\end {eqnarray}
Since for colours $c \in {\cal C}_1$, $W_0(c) + W_1(c) + W_2(c) + W_3(c) = w(vc(c))$
we get the required result.

\end{proof}

If $vc'$ is non-monochromatic, then $vc(c)$ is also non-monochromatic for any $c \in [\mu(G)]$. Therefore, $w(vc(c))=0$ for all $c$ and hence $w'(vc')=0$ from \cref{sum_eq_fin}.

If $vc'$ is monochromatic, say of colour $i$, $vc(c)$ will be  monochromatic if and only if $c=i$. Therefore, if $ i \in \mathcal{C}_2$, then $w(vc')= w(vc(i))= 1$. Similarly, if $i \in \mathcal{C}_1$, then $w'(vc')=\dfrac{1}{|\mathcal{C}_1|w(\mathbf{i}_{V_2})}$ from \cref{sum_eq_fin}.  
In both the cases this will be non-zero as required. 

The weights can now be readjusted by the Scaling Lemma (\cref{scaling_generalization}) to get a GHZ graph.

\subsection{Limitations of our reduction}
A careful reader might observe that the difficulty in extending our reduction technique to cuts of larger size comes from the case when two newly introduced edges are part of the same perfect matching. For instance, for the case when there is a $4$ vertex cut $\{u_1,u_2,u_3,u_4\}$ separating $V_1, V_2$ (both of even size) in $G$, one could try to extend our ideas and capture the weights from $V_2\cup \{u_3,u_4\}$ and $V_2\cup \{u_1,u_2\}$ on the edges $(v_3,v_4), (v_1,v_2)$ of $G'$, respectively. However, this would create some extra terms in $G'$ due to the perfect matchings in which $(v_1,v_2)$ and $(v_3,v_4)$ are contained. Such terms could destroy the GHZ property of $G'$. We believe that finding a way to bypass this difficulty and finding a more general reduction will resolve Krenn and Gu's conjecture for all graphs.

\bibliography{lipics-v2021-sample-article.bib}
%\appendix

%\input{section_2}
%\appendix
%\input{graph_quantum_connect_intro}

\end{document}